\documentclass[lettersize,journal]{IEEEtran}
\ifCLASSOPTIONcompsoc
\usepackage{cite}
\else
\usepackage{cite}
\fi
\ifCLASSINFOpdf
\else
\fi
\IEEEoverridecommandlockouts

\usepackage{amsthm}
\usepackage{mathrsfs}
\usepackage{amssymb}
\usepackage[cmex10]{amsmath}
\usepackage{stfloats}
\usepackage{graphicx}
\usepackage{subfigure} 
\usepackage{tabularx}
\usepackage{epsfig,epsf,color,balance,cite}
\usepackage{verbatim}
\usepackage{url}
\usepackage{bm}
\usepackage{caption}
\usepackage{threeparttable}
\usepackage{diagbox}
\usepackage{multirow}
\usepackage{stfloats}
\usepackage[breaklinks,colorlinks,linkcolor=black,citecolor=black,urlcolor=black]{hyperref}
\newtheorem{theorem}{Theorem}
\newtheorem{lemma}{Lemma}

\newtheorem{coro}{Corollary}
\usepackage{algorithm}
\usepackage{algorithmic}
\usepackage{multicol}

\makeatletter

\newcommand{\Rmnum}[1]{\expandafter\@slowromancap\romannumeral #1@}
\makeatother

\usepackage{color}
\definecolor{mybackground}{RGB}{204,232,207}

	\title{Online Resource Allocation for Semantic-Aware  Edge Computing Systems}

\author{
	Yihan Cang, Ming Chen, Zhaohui Yang, Yuntao Hu, Yinlu Wang,\\ Chongwen Huang, and Zhaoyang Zhang
	\IEEEcompsocitemizethanks{
		\IEEEcompsocthanksitem Y. Cang, M. Chen, Y. Hu and Y. Wang are with the National Mobile Communications Research Laboratory, Southeast University, Nanjing 210096, China (e-mails: yhcang@seu.edu.cn, chenming@seu.edu.cn,  huyuntao@seu.edu.cn, yinluwang@seu.edu.cn).  Ming Chen is also with the Purple Mountain Laboratories, Nanjing 211100, China.
		\IEEEcompsocthanksitem
		Z. Yang, C. Huang and Z. Zhang are with College of Information Science and Electronic Engineering, Zhejiang University, Hangzhou 310027, China, and with International Joint Innovation Center, Zhejiang University, Haining 314400, China, and also with Zhejiang Provincial Key Laboratory of Info. Proc., Commun. \& Netw. (IPCAN), Hangzhou 310027, China (e-mails: yang\_zhaohui@zju.edu.cn,  chongwenhuang@zju.edu.cn, ning\_ming@zju.edu.cn).}
	\vspace{-1.5em}
}

\begin{document}
	\maketitle

 	\vspace{-0cm}
	\begin{abstract}
    Mobile edge computing (MEC) in the next generation networks will provide computation services at the network edge to enrich the capabilities of mobile  devices and lengthen their  battery lives. However, the performance of MEC cannot be guaranteed, when large size local tasks are uploaded to the server simultaneously causing network congestion. 
    As a new paradigm that focuses on transmitting the meaning of messages, semantic communications reveals the significant potential to reduce the network traffic. 
    In this paper, we propose a semantic-aware joint communication and computation resource allocation framework  for MEC systems. 
    In the considered system, random tasks arrive at each terminal device (TD), which needs to be computed locally or offloaded to the MEC server.
    To further release the transmission burden, each TD sends the small-size extracted semantic information of tasks to the server instead of the original large-size raw data.
    An optimization problem of joint semantic-aware division factor, communication
    and computation resource management is formulated.
    The problem aims to minimize the energy consumption of the whole system,
    while satisfying long-term delay and processing rate constraints.
    To solve this problem, 
    an online low-complexity algorithm is proposed.
    In particular,  Lyapunov
    optimization is utilized to decompose the original coupled long-term problem into a series of decoupled deterministic problems without
    requiring the realizations of future task arrivals and channel gains.
    Then, the block coordinate descent method and successive
    convex approximation  algorithm are adopted to solve
    the current time slot deterministic problem by observing the
    current system states. 
    Moreover, the closed-form optimal solution of each optimization variable is provided.
    Simulation results show  that the proposed algorithm yields up to $41.8\%$ energy reduction compared
    to its counterpart without semantic-aware allocation. 
	\end{abstract}
 \maketitle

	\section{Introduction}
\textcolor{black}{Due to the fast development of the Internet, the desire of numerous terminal devices for a huge amount of emerging computation  services has drastically increased the traffic of core networks.  Numerous emerging services such as virtual reality (VR), augmented reality (AR) and automatic derive put forward higher request on the next generation multiple access (NGMA)  \cite{10110013}. As a result, NGMA has to accommodate multiple devices simultaneously and satisfies devices' quality-of-service (QoS), which will pose a major challenge on next generation communication networks \cite{9133094,9693417}. 
In order to release the burden of core networks, mobile edge computing (MEC) enables the network functions at the network edge, providing users with nearby real-time computing services \cite{7931566}. In future wireless networks, MEC is considered a promising paradigm to remedy the problem of computational resource and energy shortage of mobile equipment \cite{9113305}. }
In conventional MEC scenarios, communication and computation resources are jointly allocated to improve the one-shot system utility such as energy consumption \cite{8264794}, delay \cite{7997360}, throughput \cite{8249785}, computation efficiency \cite{9771419},  considering instantaneous states of the current system. However, in practical wireless communication networks, the system states including channel gains \cite{6893054}, task arrival\cite{7842160} vary with time, requiring long-term system utility optimization scheme designs. Besides, the large volume of tasks uploading will degrade the performance of MEC networks. 

Recently, there are many contributions about dynamic resource allocation and task offloading strategies for MEC systems through using stochastic optimization methods. In \cite{7842160}, the authors minimized the sum power of the system while guaranteeing task buffer stability based on Lyapunov optimization. Both theoretical analysis and numerical results verified the tradeoff between power consumption and execution delay. 
Moreover, a dynamic throughput maximum algorithm based on perturbed Lyapunov optimization was proposed in \cite{9242286} for wireless powered MEC systems. 
Through  applying extreme value theory, the authors in \cite{9040668} and\cite{8638800} proposed dynamic task offloading and resource allocation schemes to satisfy users' ultra-reliable low-latency demands. In work \cite{9449944}, the authors proposed a joint learning and optimization framework, combining the advantages of Lyapunov optimization and deep reinforcement learning to effectively schedule task offloading decisions.  
However, the above works  \cite{7842160,9242286,9040668,8638800,9449944} all ignored the download transmission period, even though the amount of result bits in the downlink is large and downlink energy consumption is inevitable such as in online gaming, argument reality, and video delivery scenarios.  Moreover, one prominent feature of the traffic in MEC systems is asymmetric. In other words, the data sizes to be uploaded and downloaded are distinctly different, which facilitates the implementation of the time division duplexing (TDD) scheme. 
Previous works \cite{8895788,8962105,9129054} have studied the TDD configuration problems in conventional MEC systems. Considering the
dynamic MEC environment, it is essential to further investigate the management of joint time division, communication and computation resource. 

\textcolor{black}{Based on Shannon's classic information theory, masses of existing works including \cite{9449944,7553459,7997360,9771419,8249785,8986845,6893054,7842160,9242286,9328479,7956189,9435770,9040668,8638800,9036074,8895788,8761346,8962105,9129054} are dedicated to the  research on data-oriented communications. As a novel paradigm that involves the 
meaning of messages in communication, semantic communications have sparked great interest from both industry and academia turning to investigate content-oriented communications. Semantic communications have revealed the significant potential to
reduce the network traffic and thus alleviate spectrum shortage\cite{9530497,10024766,10158994}. So far, various semantic communication systems have been elaborated for distinct types of sources, including text \cite{9398576}, image \cite{8723589}, speech \cite{9450827}, etc. These works  verified that semantic transmission based communications are promising to realize a significant improvement in the transmission reliability and efficiency of semantic symbols.  Most of the  existing works \cite{9398576,8723589,9450827} studied the feasibility of semantic communications under different scenarios standing at technical levels. However, the problem of how to implement efficient resource management in  semantic communications systems also needs to be investigated so as to explore the potential of semantic communications in practical networks. 
Moreover, when the wireless network is congested,  communication loads can be converted to computation tasks with the help of semantic computation in MEC systems. Thus the optimal tradeoff between communications and computations is important, which significantly improves the QoS of terminal devices. }

\textcolor{black}{Motivated by the above observations, we attempt to investigate a novel online resource management scheme, which flexibly and opportunistically coordinates communication and computation resources for semantic-aware dynamic long-term  MEC systems. In a multiple time slots scenario,  information background knowledge that has been previously transmitted can be utilized to help implement semantic compression transmission.  
The main contributions of this paper are summarized as follows:
\begin{itemize}
	\item We investigate a novel dynamic  semantic-aware TDD MEC system. In the considered system, stochastic tasks arrive at TDs  which need to be computed either locally or remotely at the server. With the help of semantic extraction, the amount of data uploading to the MEC server is significantly reduced, and thus higher energy efficiency is achieved. Moreover,   communication loads can be converted to computation tasks with the help of semantic computation in MEC systems. Thus the optimal tradeoff between communications and computations is achieved, which significantly improves the quality-of-service (QoS) of terminal devices. 
	\item To coordinate the operations on TDs and MEC server for best long-term energy efficiency, the problem
	of joint communication and computation resource management is studied. In particular, this problem is
	constructed as an optimization framework aiming to acquire the optimal local computing rate and uploading power,
	remote computing capacity and downloading power allocation, semantic extraction factors, as well as time divisions for uplinks and downlinks with satisfying the average delay
	and processing rate requirements. The proposed optimization framework applies to general semantic tasks in the viewpoint of joint communication and computation design.
	\item A dynamic resource management algorithm is proposed to solve this formulated problem. Particularly, we first utilize Lyapunov optimization technique to transform original multiple time slots stochastic optimization problem into a series of deterministic problems. Then, the current time slot deterministic problem is solved by block coordinate descent (BCD) method and successive convex approximation (SCA) algorithm. We acquire the  closed-form optimal solution for each optimization variable with low complexity. Furthermore, theoretical analysis demonstrates that the optimal semantic extraction factor exhibits a \textit{threshold-based structure} to balance the communication and computation overheads.
\end{itemize}
Simulation results verify the outstanding performance of the proposed algorithm in terms of energy efficiency. A good tradeoff between energy consumption and delay can be obtained by the proposed algorithm.  Compared with the benchmark algorithm without semantic-aware, the proposed algorithm yields up to $41.8\%$ energy reduction. }

The rest of this paper is structured as follows. Section II elaborates the system model and problem formulation. In Section III, we utilize Lyapunov optimization to convert the stochastic optimization problem into a series of deterministic problems.  The online optimization of local computing rate and uploading power, remote computing capacity and downloading power allocation, semantic extraction factor, as well as time divisions,  are rendered in Section IV. Simulation results are presented in Section V. Section VI draws the conclusions. 

	\section{System Model and Problem Formulation}
	\subsection{Network Model}
	\begin{figure}[t]
	\centering
	\includegraphics[width=0.5\textwidth]{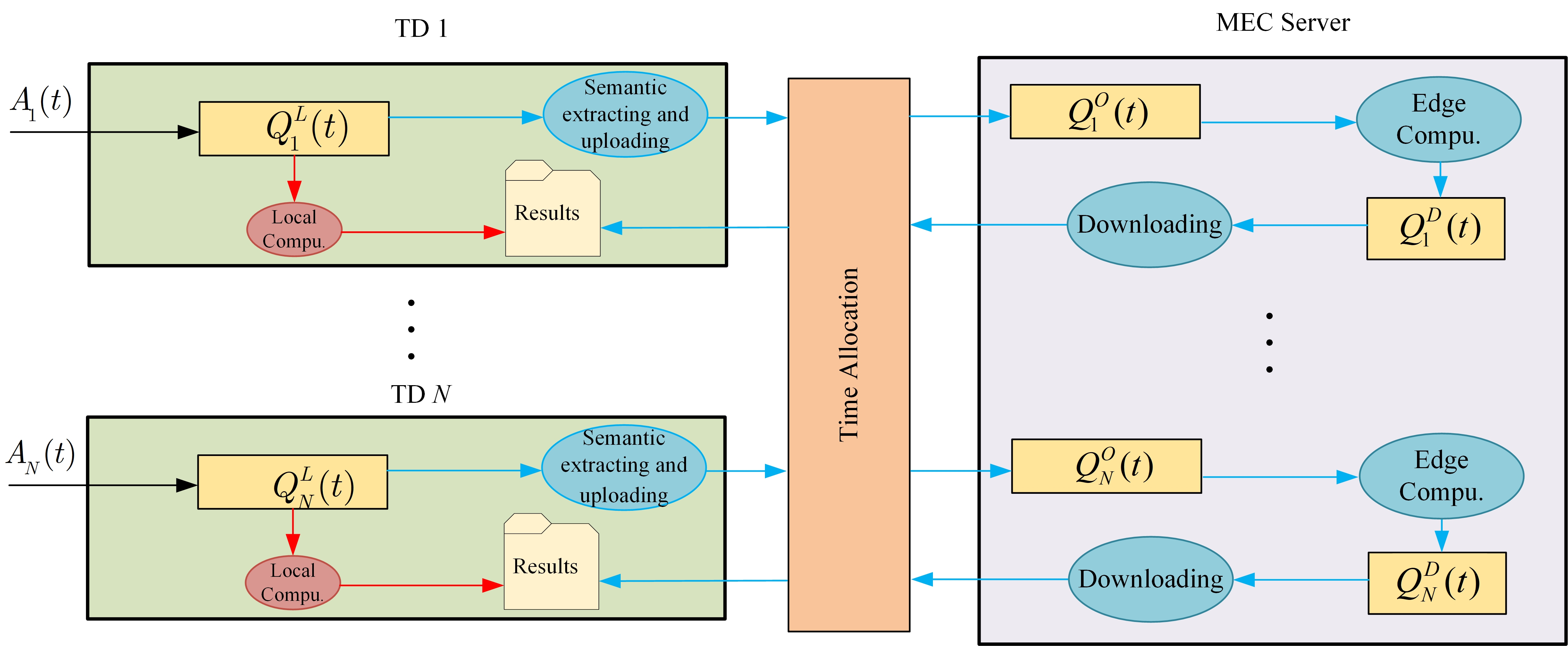}
	\caption{Flow chart of the semantic-aware MEC networks.} \label{flow chart}
	\vspace{-1.5em}
\end{figure}
	Consider a semantic-aware MEC network consisting of a set $\mathcal{N}$ of $N$ TDs with local processing capacity and a MEC server attached to a base station (BS). 
	The flow chart of the network is demonstrated in Fig.~\ref{flow chart}. 	The system contains multiple time slots and the duration of each time slot is $\tau$. The task offloading strategy  and radio resource allocation will be decided and executed at the beginning of each time slot.  Each TD has a task that arrives randomly at the beginning of each slot.  Suppose that TD $n$ has a task arriving at time slot $t$ with the size of $A_n(t)$ bits. 
	Assume the tasks are arbitrarily divisible.  We use $Q_n^L(t)$, $Q_n^{O}(t)$ and $Q_n^{D}(t)$ to respectively represent the local buffer queue,  remote processing queue, and downlink transmission queue for TD $n$ at time slot $t$ in bits. The buffers work in a First-Input First-Output (FIFO) manner.  For the local buffer in TD $n$,  the arriving task at time slot $t$ will be stored in $Q_n^L(t)$. At the same time, a part of $Q_n^L(t)$ is sent to the local computing unit and a part is offloaded to the remote processing  buffer $Q_n^O(t)$ in MEC server. The remaining tasks that have not been processed in $Q_n^L(t)$ will be executed in the several following time slots. For the remote buffers, the results of tasks in $Q_n^O(t)$ will be transferred to $Q_n^D(t)$ for downlink transmission after execution by the MEC server. Subsequently, the results of tasks are downloaded to TD $n$.  Finally, the downloaded results from $Q_n^D(t)$ as well as the results from the local computation queue $Q_n^L(t)$ are combined in local TD $n$. Without loss of generality, we set $Q_n^L(0)=Q_n^O(0)=Q_n^D(0)=0, \forall n \in\mathcal{N}$ at the initial slot.  
	
	To further decrease the transmission burden, each TD equipped with a semantic processing unit sends the  extracted semantic information of tasks to the server instead of the raw data. With the help of semantic extraction, the amount of data uploading to the MEC server is reduced, thus releasing the traffic loads.  Meanwhile, additional workloads for semantic extraction are brought to each TD and the computation intensity of tasks gets large for the MEC server in order to process semantic information. In the following subsections, we will elaborate on the detailed implementations. 
	
	\subsection{Local and Transmission Model}
	Denote $I_n$ as the computation intensity of TD $n$ in units of CPU cycles per bit. Then the local task processing rate (bits) can be obtained by
	\begin{align} \label{R_n^L(t)}
		R_n^L(t)=\frac{f_n^L(t)}{I_n},\quad \forall n \in\mathcal{N}, 
	\end{align}
	where $f_n^L(t)$ represents the local CPU frequency (cycles/second) of TD $n$ at time slot $t$. Therefore, the local computing power consumption can be formulated by $P_n^L(t)=\kappa_n f_n^L(t)^3$,  
	where $\kappa_n$ denotes the energy coefficient of TD $n$. 
	
	Orthogonal frequency division multiple access (OFDMA) is utilized for multi-user multiple access.  There are $N$ sub-channels and each TD occupies one sub-channel with bandwidth $B$.  
	Assume the channel remains static in each time slot duration but varies across different time slots. 
	\textcolor{black}{The path-loss model is given by $\bar{h}_n=A\left(\frac{c}{4\pi f_c d_n}\right)^{\ell}$, where $A$ represents antenna gain, $c$ denotes the speed of light,  $f_c$ is the carrier frequency, $\ell$ denotes the path-loss factor, and $d_n$ represents the distance between TD $n$ and the server \cite{9449944}. The instant channel gain between TD $n$ and MEC server at time slot $t$, denoted by $h_{n}(t)$, follows an i.i.d. Rician distribution with line-of-sight (LoS) 
	link gain equal to $\gamma\bar{h}_n$, where $\gamma$ is the LoS link factor. }
	Therefore, the achievable rate for uplink can be given by 
	\begin{align} \label{R_n^U(t)}
		R_{n}^U(t)=B\log_2\left (1+\frac{h_{n}(t)p_{n}^U(t)}{\sigma^2}\right),\quad \forall n \in\mathcal{N}, 
	\end{align}
	where $\sigma^2$ denotes the additive noise power at the MEC server and $p_{n}^U(t)$ represents the transmit power of TD $n$ at time slot $t$. 
	
	\begin{figure}[t]
		\centering
		\includegraphics[width=0.5\textwidth]{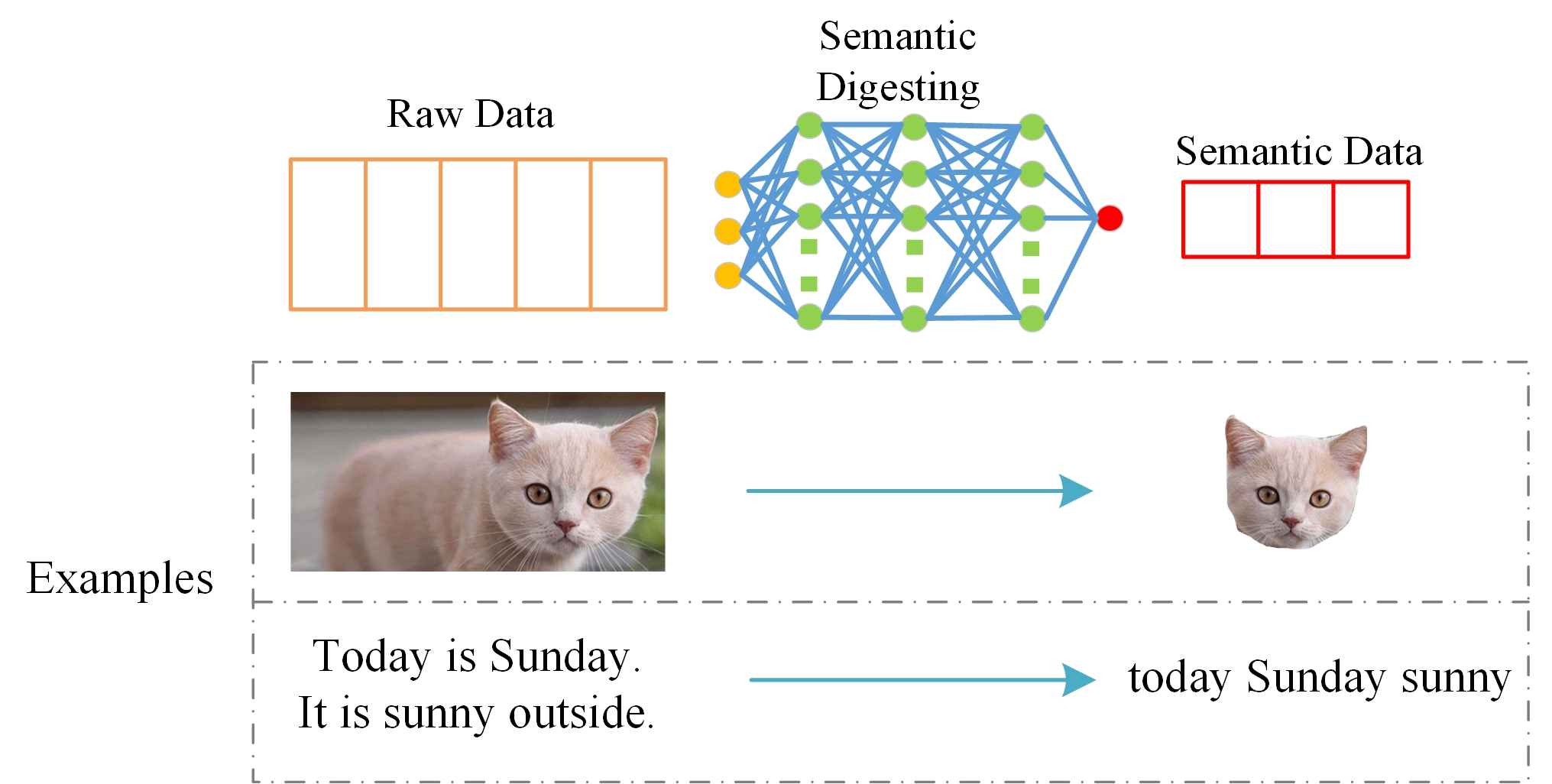}
		\caption{Schematic  diagram of semantic extracting and two examples.} \label{semantic}
		\vspace{-1.5em}
	\end{figure}
	
	 \begin{figure} 
		\centering 
		\subfigure[Accuracy versus extraction factor]{\label{}
			\includegraphics[width=0.5\linewidth]{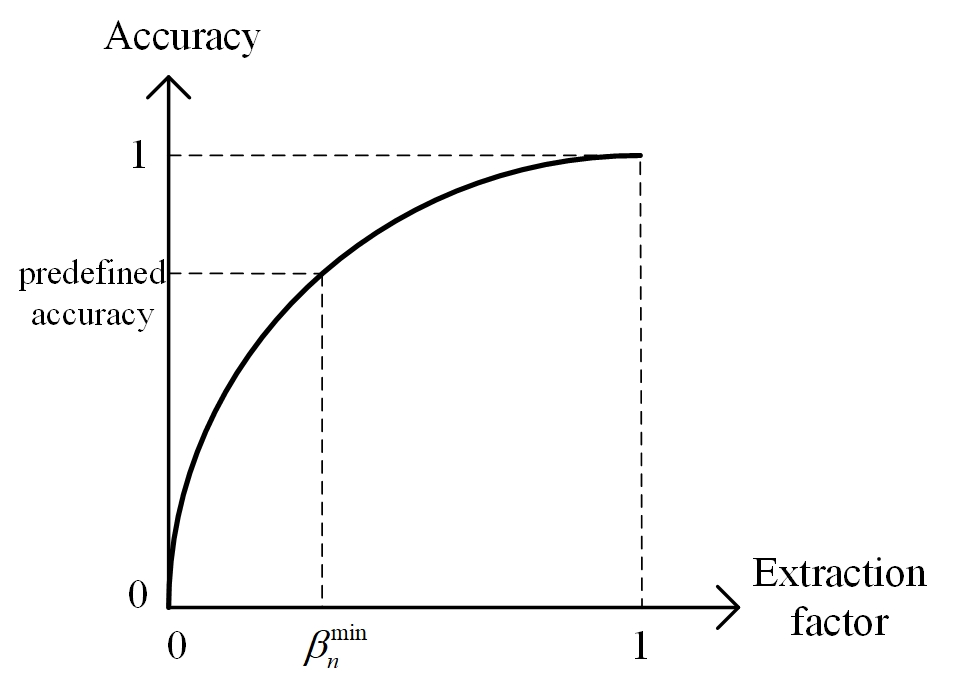}}
		\subfigure[Additional computation overhead $G_n$ versus extraction factor]{\label{}
			\includegraphics[width=0.4\linewidth]{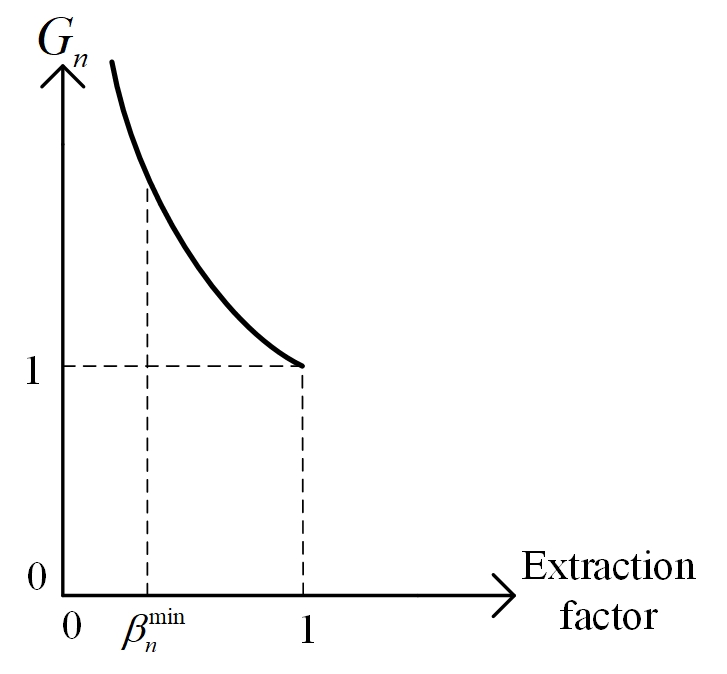}}
		\caption{Performance demonstration of semantic aware model.}
		\label{Semantic fig}
	\end{figure}
	
		\begin{figure}[t]
		\centering
		\includegraphics[width=0.5\textwidth]{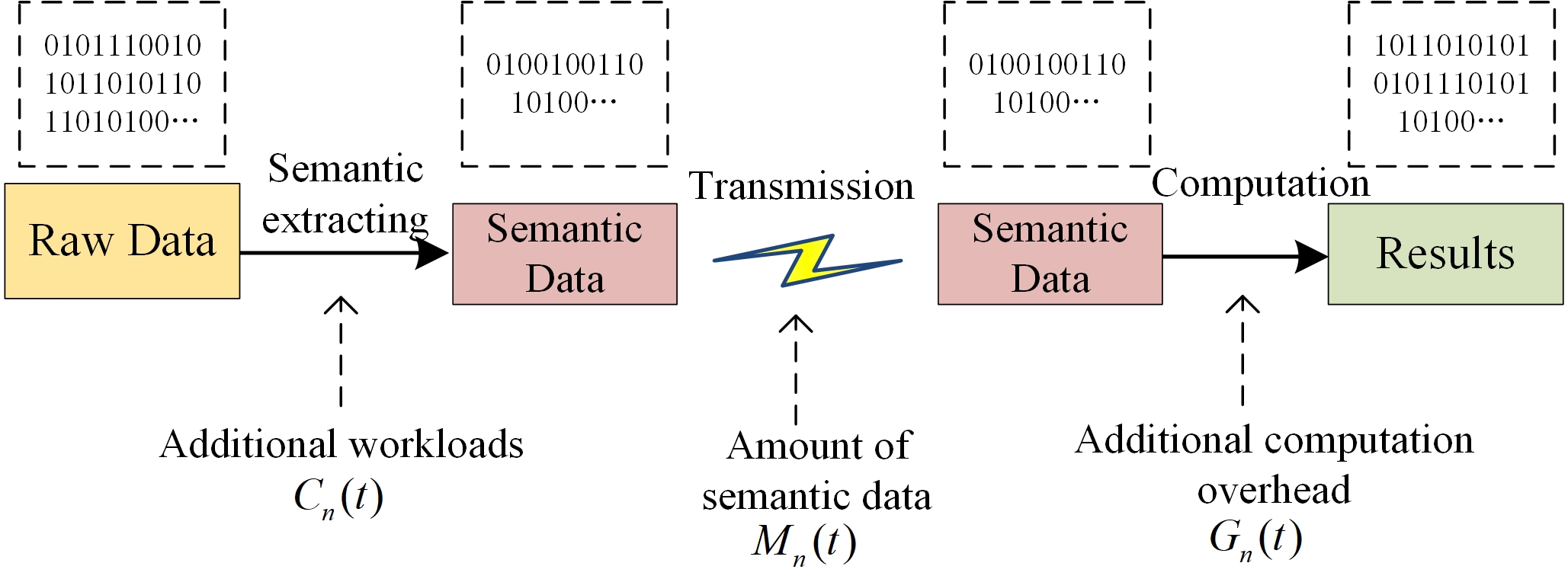}
		\caption{Schematic  diagram of semantic transmissions.} \label{semantic2}
		\vspace{-1.5em}
	\end{figure}
	
	\subsection{Semantic-Aware Model}
	\textcolor{black}{In order to further reduce offloading traffic, compute-then-transmit protocol based semantic-aware transmission is adopted, as shown in Fig.~\ref{semantic}. In specific, TD extracts and transmits the semantic information to the BS instead of raw data. For example, we can utilize the picture cutout method for image transmission, where we use a pre-trained deep neural network to cut out the part of important character while discarding the irrelevant background\cite{9274895}. For text transmission, we can discard meaningless words and use abbreviations to replace the raw text without losing valid information by constructing knowledge graphs\cite{9416312}. In the considered scenario, only simple background knowledge is required to implement semantic extraction.  Hence, the pre-trained model is lightweight and universal.   Compared with task processing,  computation resource required for semantic extraction is much less \cite{10032275,cang2023resource}.   }
	
	\textcolor{black}{Different from traditional MEC networks, we attempt to implement semantic-aware offloading scheme by allowing for a certain degree of inaccuracy. Denote $\beta_n(t)\in[\beta_n^{\min},1]$ as the extracting factor of the raw data to be uploaded by TD $n$ at time slot $t$, where $\beta_n^{\min}$ is the minimum extracting factor to maintain the integrity of TD $n$. As shown in Fig.~\ref{Semantic fig}(a), the accuracy of results by processing semantic information decreases as the extraction factor becomes small. In order to guarantee the predefined accuracy, $\beta_n(t)$ should not be smaller than $\beta_n^{\min}$. Therefore, the rate of semantics to be uploaded at time slot $t$ for TD $n$ is equivalent to $R_n^U(t)$ in equation \eqref{R_n^U(t)}. The output rate of $Q_n^L(t)$ regarding task offloading can be expressed as $R_n^U(t)/\beta_n(t)$. 	
	The amount of semantic data to be uploaded at time slot $t$ is given by
	\begin{align}
		M_n(t)=\tau_n^U(t)R_n^U(t),\quad\forall n,
	\end{align}
	where $\tau_n^U(t)$ denotes the time duration for uplink transmission of TD $n$ at time slot $t$. }
	
	\textcolor{black}{Fig.~\ref{semantic2} depicts 
	the schematic diagram of semantic transmissions. At the transmitter, semantic extraction brings additional workloads. Meanwhile, at the receiver, additional computation overhead is required to process the semantic data rather than raw data. This indicates that  communication loads are converted to computation amounts by semantic transmission. 
	Denote $C_n\left(t\right)$ as the additional workloads for semantic extraction. 
	Since small extraction factor means high compression rate and also introduces more computation load for semantic information reconstruction, thus the workload $C_n(t)$ decreases with extraction factor $\beta_n(t)$ and increases with the amount of semantic information bits $\tau_n^U(t) R_n^U (t)$. As a result, the expression of $C_n(t)$ can be expressed by 
	\begin{align} \label{5}
		C_n(t)=\frac{a\tau_n^U(t)R_n^U(t)}{\beta_n(t)^k},\quad\forall n,
	\end{align}
 	where $a>0$ and $k>1$ are constant parameters. The values of $a$ and $k$ can be obtained in simulations for specific tasks.  }

	\textcolor{black}{Since only extracted semantic information of raw data are received by the MEC server, the computation intensity of the extracted semantic increases accordingly. As shown in Fig.~\ref{Semantic fig}(b), $G_n\geq 1$, which denotes the ratio of computation intensity of semantic data to that of raw data about TD $n$, increases monotonously with extraction factor $\beta_n(t)$ getting small. The increase is caused by computations for processing semantic data and compensations for enhancing accuracy.  Besides, we let $G_n(t)=1$ in case of raw data are processed, i.e., $\beta_n(t)=1$. Without loss of generality, we assume that \footnote{This formulation can represent a kind of convex function.}
	\begin{align} \label{6}
		G_n(t)=\frac{1}{\beta_n(t)^p}, \quad\forall n,
	\end{align}
	where $p>0$ is a constant relevant to tasks. Note that $C_n(t)$, $G_n(t)$ as well as accuracy can be obtained by fitting 
	the corresponding data points of abundant prior experiments. We will show that the proposed framework is well applicable to the generality of $C_n(t)$ and $G_n(t)$ in the simulation section. }


	\vspace{-1.em}
\subsection{MEC Model}
	Denote $F_{MEC}$ as the maximum computation capacity of MEC server, then we have
	$\sum_{n=1}^N f_n^{O}(t)\leq F_{MEC}$, 
	where $f_n^O(t)$ denotes the computation resource of MEC server allocated to TD $n$ at time slot $t$. 
	Therefore, the remote computing power consumption can be formulated by
	$P_n^M(t)=\kappa_M f_n^O(t)^3$, 
	where $\kappa_M$ denotes the energy coefficient of the MEC server. 
	
	 Due to the fact that the semantic data stored in $Q_n^O(t)$ may be extracted from different factors, the semantic data to be processed at the MEC server at time slot $t$ has different extraction factors. Denote $\pmb{\chi}_n(t)$ as the aggregation of all the extraction factors of TD $n$'s data to be processed at the MEC server at current time slot $t$. For the sake of implementation, the extraction factors of the data to be processed at the MEC server at time slot $t$ share a common extraction factor when conducting remote computing, which is the minimum factor of all the data, i.e., $\min\{\pmb{\chi}_n(t)\}$. Here is an example, assume the extraction factors of  semantic data to be processed in MEC server for TD $n$ at time slot $t$ includes $\beta_n(t-3)=0.6$, $\beta_n(t-2)=0.5$, $\beta_n(t-1)=0.8$. Then the common extraction factor for TD $n$ at time slot $t$ is $\min\{0.6,0.5,0.8\}=0.5$.  
	Thus the remote semantic processing rate can be given by
	\begin{align} \label{R_n^M(t)}
	  R_n^M(t)=\frac{f_n^O(t)}{G_n(t)I_n},\quad \forall n\in\mathcal{N},
	\end{align}
	where $G_n(t)$, which is the ratio of computation intensity of semantic data to that of raw data about TD $n$ as shown in \eqref{6}, is modified as $G_n(t)=\frac{1}{\left(\min\{\pmb{\chi}_n(t)\}\right)^p}$. 
	
	Furthermore, denote $H_n(t)$ as the ratio of the amount of results to that of processed semantic data, which decreases with $\min\{\pmb{\chi}_n(t)\}$. Assume $H_n(t)=\frac{U}{\min\{\pmb{\chi}_n(t)\}}, (U>0)$ without loss of generality, where $U$ denotes the ratio of the amount of results to that of raw data.  For example, we have $H_n(t)=U$ when raw data are processed, i.e.,  $\min\{\pmb{\chi}_n(t)\}=1$.  Therefore, the input rate of $Q_n^D(t)$ can be represented by $H_n(t)R_n^M(t)$.

	Similarly, denote $p_n^D(t)$ as the downlink transmit power of BS for TD $n$ at time slot $t$. Hence, the achievable downlink rate can be given by
	\begin{align} \label{R_n^D(t)}
		R_n^D(t)=B\log_2\left(1+\frac{h_n(t)p_n^D(t)}{\sigma^2}\right),\quad \forall n \in\mathcal{N}.
	\end{align} 

		Therefore, the state transition of local buffer $Q_n^L(t)$ can be given by
	\begin{align}  \label{Q_n^L(t+1)}
		&Q_n^L(t+1)=A_n(t)+\nonumber\\
		&\max\bigg\{0, Q_n^L(t)+\underbrace{C_n\left(t\right)}_{\text{(a)}}-\underbrace{\tau R_n^L(t)}_{\text{(b)}}-\underbrace{\frac{\tau_n^U(t)R_{n}^U(t)}{\beta_n(t)}}_{\text{(c)}}\bigg\},
	\end{align} 
	where term (a) represents the additional workloads for semantic extraction as given in \eqref{5}\footnote{Since the average delay is proportional to the average queue length in a queue system\cite{8041893}, $C_n(t)$ can also describe the delay caused by semantic extraction.}, term (b) contains the tasks amounts of local computing and semantic extracting and term (c) is the amount of extracted raw data. 
	Similarly, the state transition of the remote processing queue $Q_n^O(t)$ and downlink transmission queue $Q_n^D(t)$ can be respectively given by
	\begin{align}
		Q_n^O(t+1)&=\max\left\{0, Q_n^O(t)-\tau R_n^M(t)\right\}\nonumber\\
		&+\min \left\{\tau_n^U(t)R_{n}^U(t),Q_n^L(t)-\tau R_n^L(t)\right\},\label{Q_n^O(t+1)}\\
		Q_n^D(t+1)&=\max\left\{0,Q_n^D(t)-\tau_n^D(t)R_n^D(t)\right\}\nonumber\\
		&+H_n(t)\min\left\{\tau R_n^M(t),Q_n^O(t)\right\}.\label{Q_n^D(t+1)}
	\end{align}
	
	\subsection{Problem Formulation}
	According to the Little's law\cite{8041893}, the average delay  is proportional to the average queue length. In order to satisfy the end-to-end delay requirements, we have
	\begin{align}
		\lim_{T\rightarrow\infty}\frac{1}{T}\sum_{t=1}^T\mathbb{E}[Q_n^{total}(t)]\leq Q_n^{avg}, \quad \forall n\in\mathcal{N},
	\end{align}
	where $Q_n^{total}(t)=Q_n^L(t)+Q_n^O(t)+Q_n^D(t)$,  $Q_n^{avg}$ denotes the predefined average queue length of TD $n$, and $\mathbb{E}[\cdot]$ is the expectation operation, which is taken with respect to the tasks arrivals and channel variations. 
	Moreover, in order to guarantee the processing rate of each TD, we have 
	\begin{align}
		\lim_{T\rightarrow\infty}\!\frac{1}{T}\!\sum_{t=1}^{T}\!\mathbb{E}\!\left[\tau R_n^L(t)+\tau_n^U(t)R_{n}^U(t)/\beta_n(t)\right]\geq \tau R_n^{avg},
	\end{align}
	where $R_n^{avg}$ is the predefined average processing rate of TD $n$.
	
	In this paper, we aim at minimizing the long term energy consumption of all TDs as well as the BS. The energy consumption corresponding to TD $n$'s tasks at time slot $t$ contains four parts: local computing energy,  semantic uploading energy, remote computing energy, and results downloading energy, which can be expressed as
    $E_n(t)=\tau P_n^L(t)  + \tau_n^U(t)p_n^U(t) + \tau P_n^M(t) + \tau_n^D(t)p_n^D(t)$.


	
	Let $\pmb{\Phi}(t)=\{\pmb{f}^L(t),\pmb{p}^U(t),\pmb{f}^O(t),\pmb{p}^D(t), \pmb{\beta}(t),\pmb{\tau}^U(t),\pmb{\tau}^D(t)\}$, where $\pmb{f}^L(t)=[f_1^L(t),\cdots,f_N^L(t)]^T$ and $\pmb{f}^O(t)=[f_1^O(t),\cdots,f_N^O(t)]^T$ respectively denote local and remote computation capacity vector at time slot $t$; $\pmb{p}^U(t)=[p_1^U(t),\cdots,p_N^U(t)]^T$ and $\pmb{p}^D(t)=[p_1^D(t),\cdots,p_N^D(t)]^T$ are respectively the uplink and downlink transmit power vector at time slot $t$;   $\pmb{\beta}(t)=[\beta_1(t),\cdots,\beta_N(t)]^T$ denotes extraction factor vector of all TDs at time slot $t$; $\pmb{\tau}^U(t)$ and $\pmb{\tau}^D(t)$ respectively  denote the uplink and downlink time slot division vector. 
	Mathematically, the optimization problem can be posed as 	
	\begin{subequations} \label{P1}
		\begin{align}
			\min_{\pmb{\Phi}(t)}\  &\lim_{T\rightarrow\infty}\frac{1}{T}\sum_{t=1}^T\sum_{n=1}^N\mathbb{E}[E_n(t)],\\
			\textrm{\textrm{s.t.}}\   &\lim_{T\rightarrow\infty}\frac{1}{T}\sum_{t=1}^T\mathbb{E}[Q_n^{total}(t)]\leq Q_n^{avg}, \quad \forall n\in\mathcal{N},\\
			&\lim_{T\rightarrow\infty}\!\frac{1}{T}\!\sum_{t=1}^{T}\!\mathbb{E}\!\left[\tau R_n^L(t)\!+\!\frac{\tau_n^U(t)R_{n}^U(t)}{\beta_n(t)}\right]\!\geq\! \tau R_n^{avg}, 
			\forall n\in\mathcal{N},\\
			&0\leq f_n^L(t)\leq f_n^{max}, \quad \forall n\in\mathcal{N},t,\\
			&\sum_{n=1}^N f_n^O(t)\leq F_{MEC},\quad \forall t,\\
			&0\leq p_{n}^U(t)\leq p_{n}^{max}, \quad \forall n\in\mathcal{N},t,\\
			&\sum_{n=1}^{N}p_n^D(t)\leq P_{MEC},\quad\forall t,\\  
			&\beta_n^{\min}\leq\beta_n(t)\leq 1, \quad \forall n\in\mathcal{N}, t,\\
			&\tau_n^U(t)+\tau_n^D(t)\leq \tau,\quad\forall n\in\mathcal{N}, t,\\
			&\tau_n^U(t),\tau_n^D(t)\geq 0,\quad\forall n\in\mathcal{N}, t,
		\end{align}
	\end{subequations}
	where $f_n^{max}$ and $p_{n}^{max}$ respectively denote the maximum local computation capacity and maximum transmit power of TD $n$ and $P_{MEC}$   represents the maximum transmit power of the MEC server.
	\textcolor{black}{Different from previous works \cite{7842160,9242286,9036074,9040668,8638800}, semantic extraction factor $\beta_n(t)$ is considered in problem \eqref{P1}, which involves nonconvex objective function and constraints. Moreover, additional workloads $C_n(t)$ and computation overhead $G_n(t)$ are related to tasks, which applies to general scenarios.  Besides, joint uplink and downlink transmission optimization is investigated for a dynamic multiple time slots scenario. }
	
	\textcolor{black}{Obviously, problem \eqref{P1} is a long-term stochastic optimization problem where we have to determine computing resources allocation, transmit power, and time slot division for each entity in the system at each time slot.  
	It is difficult to solve this problem due to its inherent stochastic tasks arrivals and channel variations which are priorly unknown.  Moreover, the temporally correlated decision variables make this problem more intractable. Additionally, the resource allocation and time slot division are coupled even in a single time slot, which poses another challenge for problem \eqref{P1}. Therefore, we propose an online stochastic optimization algorithm to solve this problem hinging on Lyapunov optimization method, which will be elaborated on in the next section. }

	\section{Stochastic Optimization Algorithm}
	In this section, we  utilize Lyapunov optimization technique to transform problem \eqref{P1} into a series of deterministic subproblems at each time slot. 
	First, the average expectation constraints (\ref{P1}b) and (\ref{P1}c) involve long-term calculation, which makes it hard to directly handle problem \eqref{P1}. To make problem \eqref{P1} tractable, the concept of virtual queue is introduced. 
	In particular, denote the  virtual queues of TD $n$ corresponding to constraints (\ref{P1}b) and (\ref{P1}c) by $X_n^q(t)$ and $X_n^r(t)$, respectively, with $X_n^q(0)=0$ and $X_n^r(0)=0$. Thus, the state transitions of virtual queues $X_n^q(t)$ and $X_n^r(t)$ are given as follows:
	\begin{align}
		&X_n^q(t+1)=\max\left\{0,X_n^q(t)+Q_n^{total}(t+1)-Q_n^{avg}\right\},\label{X_n^q(t+1)}\\
		&\hspace{-0em}X_n^r(t+1)=\nonumber\\
		&\max\left\{0,X_n^r(t)-\tau R_n^L(t)-\frac{\tau_n^U(t)R_{n}^U(t)}{\beta_n(t)}+\tau R_n^{avg}\right\},\label{X_n^r(t+1)}
	\end{align} 
where \eqref{X_n^q(t+1)} and \eqref{X_n^r(t+1)} are derived from (\ref{P1}b) and (\ref{P1}c), respectively. Intuitively, in order to satisfy constraints (\ref{P1}b) and (\ref{P1}c), the length of virtual queues $X_n^q(t)$ and $X_n^r(t)$ should be as short as possible.  For rigorous analysis, we have the following lemma:
\begin{lemma} \label{lemma1}
		\emph{Constraints (\ref{P1}b) and (\ref{P1}c) are guaranteed if all  virtual queues are mean rate stable, i.e., 
		\begin{align} \label{stable}
			&\lim_{t\rightarrow\infty}\frac{\mathbb{E}[|X_n^q(t)|]}{t}=0,\quad \lim_{t\rightarrow\infty}\frac{\mathbb{E}[|X_n^r(t)|]}{t}=0, \forall n \in\mathcal{N}.
		\end{align} }
	\end{lemma}

	\begin{proof}
	Since Lemma~\ref{lemma1} can be proved by using a similar method in \cite{Neely2010Stochastic}, the proof is omitted here. 
\end{proof}
	
	Lemma~\ref{lemma1} implies that problem \eqref{P1} with time-average expectation  constraints is equivalent to the following problem with queue stability constraints:
	\begin{subequations} \label{P2}
	\begin{align}
		\max_{\pmb{\Phi}(t)}\quad &\lim_{T\rightarrow\infty}\frac{1}{T}\sum_{t=1}^T\sum_{n=1}^N\mathbb{E}[E_n(t)],\\
		\textrm{\textrm{s.t.}} \quad&  \text{(\ref{P1}d)\ -\ (\ref{P1}j)},\ \eqref{stable}. 
	\end{align}
\end{subequations} 
To this end, we employ Lyapunov optimization to solve problem \eqref{P2}. 
First, the Lyapunov function that measures the backlog of virtual queues  can be posed as
$L(\pmb{\Theta}(t))=\frac{1}{2}\sum_{n=1}^N \left[X_n^q(t)^2+X_n^r(t)^2\right]$,
where $\pmb{\Theta}(t)=[X_n^q(t),X_n^r(t)|\forall n \in\mathcal{N}]$ contains all the states of virtual queues. Thus the one-step conditional Lyapunov drift function at time slot $t$ can be given by $\Delta(\pmb{\Theta}(t))=\mathbb{E}\left[L(\pmb{\Theta}(t+1))-L(\pmb{\Theta}(t))|\pmb{\Theta}(t)\right]$. 
Obviously, minimizing $\Delta(\pmb{\Theta}(t))$ leads to the most stable virtual queues while resulting in non-ideal system energy consumption. In order to achieve a tradeoff between the queue stability and long term average energy consumption, we try to minimize the following drift-plus-penalty function at every time slot as an alternative:
\begin{align} \label{Delta_V}
	\Delta_V(\pmb{\Theta}(t))=\Delta(\pmb{\Theta}(t))+V\mathbb{E}\left[\sum_{n=1}^NE_n(t)\Big|\pmb{\Theta}(t)\right],
\end{align}
 where $V$ is a non-negative control parameter balancing queues length and system energy consumption.
  
 \begin{theorem} \label{theorem_1}
 	\emph{The upper bound of $\Delta_V(\pmb{\Theta}(t))$ is given as follows for all $V$, $t$,  and $\pmb{\Theta}(t)$:
 	\begin{small}
 	\begin{align} \label{theorem_1_eq}
 		&\Delta_V(\pmb{\Theta}(t))\leq \Xi+\Pi(t)+\sum_{n=1}^N\mathbb{E}\bigg[2Q_n^L(t)C_n(t)\nonumber\\
 		&-2\left(Q_n^L(t)+C_n^{max}\right)\left(\tau R_n^L(t)+\frac{\tau_n^U(t)R_n^U(t)}{\beta_n(t)}\right)\nonumber\\
 		&+4Q_n^O(t)\left(\tau_n^U(t)R_n^U(t)-\tau R_n^M(t)\right)\nonumber\\
 		&+4Q_n^D(t)\left(H_n(t)\tau R_n^M(t)-\tau_n^D(t)R_n^D(t)\right)\nonumber\\
 		&+X_n^q(t)\bigg(\max\bigg\{0, Q_n^L(t)+C_n\left(t\right) -\tau R_n^L(t)-\frac{\tau_n^U(t)R_{n}^U(t)}{\beta_n(t)}\bigg\}\nonumber\\
 		&+\max\left\{0, Q_n^O(t)\!-\!\tau R_n^M(t)\right\}\!+\! \min \left\{\tau_n^U(t)R_{n}^U(t),Q_n^L(t)\!-\!\tau R_n^L(t)\right\}\nonumber\\
 		&+\max\left\{0,Q_n^D(t)\!-\!\tau_n^D(t)R_n^D(t)\right\}\!+\! H_n(t)\min\left\{\tau R_n^M(t),Q_n^O(t)\right\}\bigg)\nonumber\\
 		&-X_n^r(t)\left(\tau R_n^L(t)+\frac{\tau_n^U(t)R_{n}^U(t)}{\beta_n(t)}\right)+VE_n(t)\bigg|\pmb{\Theta}(t)\bigg],
 	\end{align}
\end{small}where $\Xi$ is a constant of all time slots and $\Pi(t)$ is a constant at time slot $t$ which involves the current system states:
\begin{small}
 \begin{align}
 	\Xi=&\sum_{n=1}^N\bigg[(C_n^{max})^2+\tau^2\left(R_n^{L,max}+\frac{R_n^{U,max}}{\beta_n^{\min}}\right)^2+2\tau^2(R_n^{D,max})^2\nonumber\\
 	&+\frac{1}{2}(Q_n^{avg})^2+\frac{1}{2}\left(\tau R_n^{L,\max}+\frac{\tau R_{n}^{U,\max}}{\beta_n^{\min}}\right)^2\nonumber\\
 	&+2\tau^2(R_n^{M,max})^2+2\tau^2(R_n^{U,max})^2+\frac{1}{2}\left(\tau R_n^{avg}\right)^2\bigg].
 \end{align}\end{small}
\begin{small}
 \begin{align}
 	\Pi(t)=&\sum_{n=1}^N\bigg[Q_n^L(t)^2+A_n(t)^2+2\left(Q_n^L(t)+C_n^{max}\right)A_n(t)\nonumber\\
 	&+2Q_n^O(t)^2+2Q_n^D(t)^2+2(H_n(t)\tau R_n^{M,max})^2\nonumber\\
 	&+X_n^q(t)(A_n(t)-Q_n^{avg})+X_n^r(t)\tau R_n^{avg}\bigg].
 \end{align}\end{small}}
 \end{theorem} 
  \begin{proof}
  	Please refer to Appendix~\ref{proof_1}. 
  \end{proof}
According to  Theorem~\ref{theorem_1}, the drift-plus-penalty minimization algorithm proceeds with minimizing the right hand side (RHS) of \eqref{theorem_1_eq} opportunistically by observing the states of all the queues at each time slot. 
  However, the RHS of \eqref{theorem_1_eq} is non-differentiable due to the existence of $\max\{\cdot\}$ and  $\min\{\cdot\}$ operations. To this end, a differentiable upper bound of the RHS of \eqref{theorem_1_eq} is found according to the following steps.
  Firstly, the $\max\{0,x\}$ is equivalent to $x$ if $x\geq 0 $ holds. Therefore, the max function in the RHS of \eqref{theorem_1_eq} can be removed if 
  \begin{align} 
  	 \tau R_n^L(t)+\frac{\tau_n^U(t)R_{n}^U(t)}{\beta_n(t)}&\leq Q_n^L(t)+C_n\left(t\right),\quad  \forall n \in\mathcal{N},\label{max_appro}\\
  	 \tau R_n^M(t) &\leq Q_n^O(t),\quad \forall n \in\mathcal{N},\\
  	 \tau_n^D(t) R_n^D(t) &\leq Q_n^D(t),\quad \forall n \in\mathcal{N}, \label{max_appro3}
  \end{align}
  are added in the constraints. This scaling is reasonable on account of the causality of buffers, i.e., the amount of data to be processed at the  current time slot cannot exceed the amount of overstock data in the buffer. For the min function, we have
  \begin{align}
  	&\min \left\{\tau_n^U(t)R_{n}^U(t),Q_n^L(t)\!-\!\tau R_n^L(t)\right\}\!\leq\!\tau_n^U(t)R_{n}^U(t),\forall n \in\mathcal{N},\\
  	&\min\left\{\tau R_n^M(t),Q_n^O(t)\right\}\leq\tau R_n^M(t), \forall n \in\mathcal{N}.\label{min_appro}
  \end{align}
 
  Hence, based on \eqref{max_appro}-\eqref{min_appro}, the following problem is reconstructed with its objective function lower bounded by the RHS of \eqref{theorem_1_eq}:
  \begin{small}
\begin{subequations} \label{P4}
	\begin{align}
		\min_{\pmb{\Phi}(t)}& \ 	\sum_{n\in\mathcal{N}}\bigg[2Q_n^L(t)C_n(t)+4Q_n^O(t)\left(\tau_n^U(t)R_n^U(t)-\tau R_n^M(t)\right)\nonumber\\
		&-2\left(Q_n^L(t)+C_n^{max}\right)\left(\tau R_n^L(t)+\frac{\tau_n^U(t)R_n^U(t)}{\beta_n(t)}\right)\nonumber\\
		&+X_n^q(t)\bigg( C_n\left(t\right) -\tau R_n^L(t)-\frac{\tau_n^U(t)R_{n}^U(t)}{\beta_n(t)}-\tau R_n^M(t)\nonumber\\
		&+\tau_n^U(t)R_{n}^U(t)-\tau_n^D(t)R_n^D(t)+H_n(t)\tau R_n^M(t)\bigg)\nonumber\\
		&+4Q_n^D(t)\left(H_n(t)\tau R_n^M(t)-\tau_n^D(t)R_n^D(t)\right)\nonumber\\
		&-X_n^r(t)\left(\tau R_n^L(t)+\frac{\tau_n^U(t)R_{n}^U(t)}{\beta_n(t)}\right)+VE_n(t)\bigg],\\
		\textrm{\textrm{s.t.}}&\ \   \text{(\ref{P1}d)}-\text{(\ref{P1}j)},\ \eqref{max_appro}-\eqref{max_appro3},
	\end{align}
\end{subequations} \end{small}
where the constant terms are removed. 

Therefore, multiple time slots stochastic optimization problem \eqref{P1} can be solved by obtaining the solution of  problem \eqref{P4} once at every time slot $t$ through observing the current states of queues  $\pmb{\Theta}(t)$, channel gains $h_n(t)$, and  extraction factors of the data to be processed at MEC server $\pmb{\chi}_n(t)$.  In the next section, an efficient algorithm is proposed to solve problem \eqref{P4}. 

  	\section{Online Optimization Algorithm}
  	By observing the structure of problem \eqref{P4}, it can be decoupled into two problems, which means that we can solve  problem \eqref{P4} into two stages. In the first stage, we solve the local process and transmit power problem. In the second stage, we optimize  remote computing capacity allocation problem. 
  	
  	\subsection{Local Process and Transmission Power Problem}
  	According to \eqref{P4}, the local process and transmit power problem is formulated as following:
  	\begin{subequations} \label{P5}
  	 \begin{align}
  	 	&\min_{\overset{\pmb{f}^L(t),\pmb{p}^U(t),\pmb{p}^D(t),}{\pmb{\beta}(t),\pmb{\tau}^U(t),\pmb{\tau}^D(t)}}  \sum_{n=1}^N\bigg[\tilde{W}_n^1(t)C_n(t)+\tilde{W}_n^3(t)\tau_n^U(t)R_n^U(t)\nonumber\\
  	 	&-\tilde{W}_n^2(t)\left(\tau\frac{f_n^L(t)}{I_n}+\frac{\tau_n^U(t) R_n^U(t)}{\beta_n(t)}\right)-\tilde{W}_n^4(t)\tau_n^D(t)R_n^D(t)\nonumber\\
  	 	&+V\left(\tau\kappa_n f_n^L(t)^3+\tau_n^U(t)p_n^U(t)+\tau_n^D(t)p_n^D(t)\right)\bigg]\\
  	 	&\textrm{\textrm{s.t.}}\quad  \eqref{max_appro}, \eqref{max_appro3}, \text{(\ref{P1}d)}, \text{(\ref{P1}f)-}\text{(\ref{P1}j)},
  	 \end{align}
  	\end{subequations}
  	where $\tilde{W}_n^1(t)=2Q_n^L(t)+X_n^q(t)\geq 0$,  $\tilde{W}_n^2(t)=X_n^r(t)+2Q_n^L(t)+2C_n^{\max}+X_n^q(t)>0$, $\tilde{W}_n^3(t)=4Q_n^O(t)+X_n^q(t)\geq 0$, $\tilde{W}_n^4(t)=4Q_n^D(t)+X_n^q(t)\geq 0$. 
  	Problem \eqref{P5} is non-convex. Moreover, due to the coupling relationship between transmit power, extraction factor, and time division, it is quite difficult to solve this problem. To this end, block coordinate descent (BCD) method is adopted to solve problem \eqref{P5}. 
  	  	
  	\subsubsection{Local Computing Rate and Transmission Power Optimization}
  	We first optimize local computing rate and transmit power with fixed extraction factor and time division.  However, we cannot directly deduce whether the coefficient of  
  	$R_n^U(t)$ is positive or not in  (\ref{P5}a). 
  	Therefore, the concavity or convexity of the objective function of \eqref{P5} with respect to $p_n^U(t)$ cannot be judged even though $R_n^U(t)$ is concave with $p_n^U(t)$. 
  	Hence, we optimize $r_n^U(t)$, which is the uplink transmission rate for TD $n$ at time slot $t$, instead of optimizing $p_n^U(t)$ directly.
  	According to \eqref{R_n^U(t)}, we can obtain 
  	$p_n^U(t)=\frac{\sigma^2}{h_n(t)}\left(e^{\frac{r_n^U(t)\ln 2}{B}}-1\right)$, 
which is convex with respect to $r_n^L(t)$. With given extraction factor and time division variable, problem \eqref{P5} can be reformulated as 
  	\begin{subequations} \label{32}
	\begin{align}
		\hspace{-1em}\min_{\overset{\pmb{f}^L(t),\pmb{r}^U(t),}{\pmb{p}^D(t)}} & \sum_{n=1}^N \bigg[ -\frac{\tilde{W}_n^2(t)\tau}{I_n}f_n^L(t)-\tilde{W}_n^4(t)\tau_n^D(t)R_n^D(t)\nonumber\\
		&+\left(\frac{a\tilde{W}_n^1(t)}{\beta_n(t)^k}-\frac{\tilde{W}_n^2(t)}{\beta_n(t)}+\tilde{W}_n^3(t)\right)\tau_n^U(t) r_n^U(t)\nonumber\\
		&+V\left(\tau\kappa_n f_n^L(t)^3+\tau_n^U(t)\frac{\sigma^2}{h_n(t)}\left(e^{\frac{r_n^U(t)\ln 2}{B}}-1\right)\right.\nonumber\\
		&+\tau_n^D(t)p_n^D(t)\bigg)\bigg]\\
		\textrm{\textrm{s.t.}} \quad  & \tau\frac{f_n^L(t)}{I_n}\!+\!\frac{\tau_n^U(t)r_{n}^U(t)}{\beta_n(t)}\!\leq\! Q_n^L(t)\!+\!\!\frac{a\tau_n^U(t)r_n^U(t)}{\beta_n(t)^k}\!, \forall n \!\in\!\!\mathcal{N}\!,\\
		&p_n^D(t)\leq \frac{\sigma^2}{h_n(t)}\left(e^{\frac{Q_n^D(t)\ln 2}{B\tau_n^D(t)}}-1\right), \forall n \in\mathcal{N},\\
		&0\leq f_n^L(t)\leq f_n^{\max},\quad \forall n \in\mathcal{N},\\
		&r_n^U(t)\leq B\log_2 \left(1+\frac{h_n(t)p_n^{\max}}{\sigma^2}\right) , \forall n \in\mathcal{N},\\
		&\sum_{n=1}^{N}p_n^D(t)\leq P_{MEC},
	\end{align}
\end{subequations}
	where $\pmb{r}^U(t)=[r_1^U(t),\cdots,r_N^U(t)]^T$. Problem \eqref{32} is convex and the optimal solution can be obtained in a closed form as given in the following corollary.
	\begin{coro} \label{coro1}
		\emph{The optimal solution of \eqref{32} is given by
		\begin{align} 
			p_n^D(t)&=\min\left\{\max\left\{0,\Omega_1\right\},\frac{\sigma^2}{h_n(t)}\left(e^{\frac{Q_n^D(t)\ln 2}{B\tau_n^D(t)}}-1\right)\right\},\label{33}\\
		f_n^L(t)&=\left\{\begin{aligned}
			0,\quad\quad\quad\quad\quad\quad\quad &\text{if } \rho_n>\tilde W_2(t);\\
			\min\left\{\Omega_2,f_n^{\max}\right\},\quad &\text{otherwise},
		\end{aligned}
		\right.\label{34}\\
	r_n^U(t)&=\left\{\begin{aligned}
		&0,\quad\quad\quad\quad\quad\quad\quad\quad\quad\quad\quad\quad\quad\quad \text{if } \omega_n(t)\geq0;\\
		&\min\left\{\Omega_3, B\log_2\left(1+\frac{h_n(t)p_n^{\max}}{\sigma^2}\right)\right\}, \text{otherwise},
	\end{aligned}\label{35}
	\right.
\end{align}
	where $\mu$ and $\rho_n$ are non-negative Lagrangian multipliers which can be obtained by bisection method as shown in Algorithm~\ref{algB},  $\Omega_1=\frac{\tilde{W}_n^4(t)\tau_n^D(t)B}{\ln 2\left(V\tau_n^D(t)+\mu\right)}-\frac{\sigma^2}{h_n(t)}$, $\Omega_2=\sqrt{\frac{\tilde{W}_n^2(t)-\rho_n}{3VI_n\kappa_n}}$, $\Omega_3=B\log_2 \frac{-\omega_n(t)Bh_n(t)}{\tau_n^U(t)\ln2V\sigma^2}$, 
	$\omega_n(t)=\left(\frac{a\tilde{W}_n^1(t)}{\beta_n(t)^k}-\frac{\tilde{W}_n^2(t)}{\beta_n(t)}+\tilde{W}_n^3(t)\right)\tau_n^U(t)+\rho_n\tau_n^U(t)\left(\frac{1}{\beta_n(t)}-\frac{a}{\beta_n(t)^k}\right)$. }
	\end{coro}
\begin{proof}
	Please refer to Appendix B.
\end{proof}

\begin{algorithm}[t]   
	\algsetup{linenosize= \footnotesize}
	 \footnotesize
	\caption{\textcolor{black}{Optimal Local Computing Rate and Transmission Power Algorithm}}
	\begin{algorithmic}[1]         \label{algB}
		\STATE {\bf Initialization:}  $\mu_{ub}\leftarrow$ sufficiently large,  $\mu_{lb}\leftarrow 0$, required precision $\epsilon_1,\epsilon_2$. 
		
		\STATE  {\bf Repeat:}
		\STATE   \hspace*{0.2in}  Set $\mu \leftarrow \frac{\mu_{lb}+\mu_{ub}}{2}$;\\
		\STATE \hspace*{0.2in} Obtain $p_n^D(t)$ according to \eqref{33}. 
		\STATE \hspace*{0.2in} {\bf If} $\sum_{n=1}^N p_n^D(t)- P_{MEC} < 0$:\\
		\hspace*{0.4in} $\mu_{ub}\leftarrow\mu$;\\
		\hspace*{0.2in} {\bf Else:}\\
		\hspace*{0.4in} $\mu_{lb}\leftarrow\mu$.\\
		\hspace*{0.2in} {\bf End}\\
		{\bf Until:} $\mu_{ub}-\mu_{lb}\leq\epsilon_1$. 
		\STATE {\bf For} $n=1:N$:
		\STATE \hspace*{0.2in} Set $\rho_n^{ub}\leftarrow$ sufficient large, $\rho_n^{lb}\leftarrow 0$,
		\STATE \hspace*{0.2in} {\bf Repeat:}
		\STATE \hspace*{0.4in} Set $\rho_n=(\rho_n^{ub}+\rho_n^{lb})/2$,\\
		\STATE\hspace*{0.4in}  Obtain $f_n^L(t)$ and $r_n^U(t)$ according to \eqref{34} and \eqref{35}, respectively. 
		
		\STATE \hspace*{0.4in} {\bf If} $\tau\frac{f_n^L(t)}{I_n}+\frac{\tau_n^U(t)r_{n}^U(t)}{\beta_n(t)}- Q_n^L(t)-\frac{a\tau_n^U(t)r_n^U(t)}{\beta_n(t)^k}<0$:\\
		\hspace*{0.6in} $\rho_n^{ub}\leftarrow\rho_n$;\\
		\hspace*{0.4in} {\bf Else:}\\
		\hspace*{0.6in} $\rho_n^{lb}\leftarrow\rho_n$.\\
		\hspace*{0.4in} {\bf End}\\
		\hspace*{0.2in} {\bf Until:} $\rho_n^{ub}-\rho_n^{lb}\leq\epsilon_2$.\\
		{\bf End}
		
		\STATE  {\bf Output:} optimal $p_n^D(t)$, $f_n^L(t)$ and $r_n^U(t)$. 
	\end{algorithmic}
\end{algorithm}

\textcolor{black}{
	The optimal local computing rate and transmit power algorithm flow is summarized in Algorithm~\ref{algB}. From Steps 2 to 5, the optimal $\mu$ is obtained by bisection method and the optimal downlink transmit power is then obtained according to \eqref{33}. Subsequently, from Steps 6 to 11, the optimal $\{\rho_n\}$ are solved in a parallel manner and then the optimal local computing rate and uplink transmission rate are obtained respectively according to \eqref{34} and \eqref{35}. }

\textit{Remark:} According to Corollary \ref{coro1}, the optimal $p_n^D(t)$ increases with $Q_n^D(t)$ because when the length of the downlink transmission queue is large, more transmit power is required in order to keep the downlink queue stable. The optimal $f_n^L(t)$ and $r_n^U(t)$ both increase with $Q_n^L(t)$ as it is desirable to enlarge local processing rate and uploading rate to keep the length of the local queue small. Moreover, $f_n^L(t)$ and $r_n^U(t)$ increase with $X_n^r(t)$. This can be explained by that a large $X_n^r(t)$ indicates the processing rate is lower than required. Thus, we have to increase $f_n^L(t)$ and $r_n^U(t)$ so as to enlarge the processing rate. Furthermore, as $V$ increases, $p_n^D(t)$, $f_n^L(t)$ and $r_n^U(t)$ all becomes small. This is due to the fact that when $V$ is large, more weights are put on energy consumption than queue length. Therefore, smaller $p_n^D(t)$, $f_n^L(t)$ and $r_n^U(t)$ are preferred to bring in less energy consumption.

  	\subsubsection{Extraction factor Optimization}
  	The extraction factor optimization problem can be decomposed into $N$ subproblems with the same structure, which can be solved in a parallel manner. With given local computing rate, transmit power, and time division variable, problem \eqref{P5} can be decomposed as    
  	 	\begin{subequations} \label{Digestion factor Optimization}
  		\begin{align}
  			\hspace{-0em}\min_{\beta_n(t)}\  & \tau_n^U(t) r_n^U(t)\left(\frac{a\tilde{W}_n^1(t)}{\beta_n(t)^k}-\frac{\tilde{W}_n^2(t)}{\beta_n(t)}\right)\\
  			\textrm{s.t.} \  & \tau\frac{f_n^L(t)}{I_n}\!+\!\frac{\tau_n^U(t)r_{n}^U(t)}{\beta_n(t)}\!\leq\! Q_n^L(t)\!+\!\frac{a\tau_n^U(t)r_n^U(t)}{\beta_n(t)^k},\\
  			&\beta_n^{\min}\leq \beta_n(t)\leq 1.
  		\end{align}
  	\end{subequations}
  Since  the convexity of (\ref{Digestion factor Optimization}a) with respect to $\beta_n(t)$ cannot be guaranteed, through substituting $\tilde{\beta}_n(t)=\frac{1}{\beta_n(t)}$ for $\beta_n(t)$, problem (\ref{Digestion factor Optimization}) can be rewritten as
 \begin{subequations} \label{Digestion factor Optimization1}
 	\begin{align}
 		\hspace{-0.5em}\min_{\tilde \beta_n(t)}\  & \tau_n^U(t) r_n^U(t)\left(a\tilde{W}_n^1(t)\tilde \beta_n(t)^k-\tilde{W}_n^2(t)\tilde\beta_n(t)\right)\\
 		\textrm{\textrm{s.t.}} \  & \tau\frac{f_n^L(t)}{I_n}\!+\!\tau_n^U(t)r_{n}^U(t)\tilde\beta_n(t)\!\leq\! Q_n^L(t)\!+\!a\tau_n^U(t)r_n^U(t)\tilde\beta_n(t)^k\!,\\
 		&1\leq \tilde\beta_n(t)\leq \frac{1}{\beta_n^{\min}}.
 	\end{align}
 \end{subequations}
  whose objective function is convex now. However, RHS of (\ref{Digestion factor Optimization1}b) is non-concave. To this end, we apply first-order Taylor's formula to approximate $\tilde\beta_n(t)^k$ as
  \begin{align} \label{taylor}
  	\tilde\beta_n(t)^k\geq k\left(\tilde{\beta}_n^{(r)}(t)\right)^{k-1}\tilde\beta_n(t)+(1-k)\left(\tilde\beta_n^{(r)}(t)\right)^k. 
  \end{align}
  Through plugging \eqref{taylor} into (\ref{Digestion factor Optimization1}b), problem \eqref{Digestion factor Optimization1} is reformulated as
   \begin{small}\begin{subequations} \label{Digestion factor Optimization2}
   	\begin{align}
   		\min_{\tilde \beta_n(t)} \  & \tau_n^U(t) r_n^U(t)\left(a\tilde{W}_n^1(t)\tilde \beta_n(t)^k-\tilde{W}_n^2(t)\tilde\beta_n(t)\right)\\
   		\textrm{\textrm{s.t.}}\  & \tau\frac{f_n^L(t)}{I_n}+\tau_n^U(t)r_{n}^U(t)\tilde\beta_n(t)\leq Q_n^L(t)+a\tau_n^U(t)r_n^U(t)\nonumber\\
   		&\times\left[k\left(\tilde{\beta}_n^{(r)}(t)\right)^{k-1}\tilde\beta_n(t)+(1-k)\left(\tilde\beta_n^{(r)}(t)\right)^k\right],\\
   		&1\leq \tilde\beta_n(t)\leq \frac{1}{\beta_n^{\min}},
   	\end{align}
   \end{subequations}\end{small}which is convex now. Through solving problem \eqref{Digestion factor Optimization2}, we have the following theorem. 
  \begin{theorem} \label{coro2}
  	\emph{The optimal $\beta_n(t)$ exhibits a  \textit{threshold-based} structure
  	\begin{align} \label{coro2eq}
  		\beta_n(t)\!=\!\left\{
  		\begin{aligned}
  			&1,\  \text{if } \tau_n^U(t) r_n^U(t) \!\left[\xi\!-\!\tilde{W}_n^2(t)\!-\!\xi a k\left(\tilde{\beta}_n^{(r)}(t)\right)^{k-1}\!\right]\!\!\geq\! 0 \\
  			&\quad \text{ or } \Omega_4<1;\\
  			&\max\left\{\frac{1}{\Omega_4},\beta_n^{\min}\right\},\quad  \text{otherwise}.
  		\end{aligned}\right.
  	\end{align}
  where $\xi\geq 0$ is dual variable with respect to (\ref{Digestion factor Optimization2}b) and   $\Omega_4=\sqrt[k-1]{\frac{\tilde{W}_n^2(t)-\xi+\xi a k\left(\tilde{\beta}_n^{(r)}(t)\right)^{k-1}}{ka\tilde{W}_n^1(t)}}$.}
  \end{theorem}
  \begin{proof}
  	Please refer to Appendix C.
  \end{proof}

\textit{Remark:} According to Theorem \ref{coro2}, the optimal extraction factor $\beta_n(t)$ decreases as $X_n^r(t)$ and $X_n^q(t)$ becomes large. This is because a longer queue $X_n^r(t)$ implies that system processing rate is lower than required. Thus, reducing the extraction factor can enlarge the processing rate. Meanwhile,  a longer queue $X_n^q(t)$ indicates that the data queue is longer than required. Therefore, we can also reduce the extraction factor to keep the data queue shorter. 

  	\subsubsection{Time Division Optimization}
  	With given local computing rate, transmit power, and extraction factor  variable in problem \eqref{P5}, time division optimization problem can be decomposed into the following $N$ parallel subproblems: 
  	  	\begin{small}\begin{subequations} \label{time division1}
  		\begin{align}
  			&\hspace{-2em}\min_{\tau_n^U(t),\tau_n^D(t)}\  \left(Vp_n^D(t)-\tilde{W}_n^4(t)R_n^D(t)\right)\tau_n^D(t)\nonumber\\
  			&\hspace{-2em}+\bigg(\frac{\tilde{W}_n^1(t)a r_n^U(t)}{\beta_n(t)^k}-\frac{\tilde{W}_n^2(t)r_n^U(t)}{\beta_n(t)}+\tilde{W}_n^3(t)r_n^U(t)+Vp_n^U(t)\bigg)\tau_n^U(t)\\ 
  			\textrm{\textrm{s.t.}} \ &   \tau\frac{f_n^L(t)}{I_n}+\frac{\tau_n^U(t)R_{n}^U(t)}{\beta_n(t)}\leq Q_n^L(t)+\frac{a\tau_n^U(t)R_n^U(t)}{\beta_n(t)^k},\\
  			&\tau_n^D(t) R_n^D(t) \leq Q_n^D(t),\\  
  			&\tau_n^U(t)+\tau_n^D(t)\leq \tau,\\
  			&\tau_n^U(t),\tau_n^D(t)\geq 0,
  		\end{align}
  	\end{subequations}\end{small}which is a linear programming problem.  
  	Denote $\xi^U=\frac{\tilde{W}_n^1(t)a r_n^U(t)}{\beta_n(t)^k}-\frac{\tilde{W}_n^2(t)r_n^U(t)}{\beta_n(t)}+\tilde{W}_n^3(t)r_n^U(t)+Vp_n^U(t)$ and $\xi^D=Vp_n^D(t)-\tilde{W}_n^4(t)R_n^D(t)$ as the coefficients of $\tau_n^U(t)$ and $\tau_n^D(t)$, respectively. Due to the uncertainty of whether constraint (\ref{time division1}b) restrict the upper bound or lower bound of $\tau_n^U(t)$, we need to consider the following two cases.
  	
  	Case 1: $\frac{1}{\beta_n(t)}-\frac{a}{\beta_n(t)^k}\geq 0$. In this case, $\tau_n^U(t)$ is upper bounded by $\theta=\frac{Q_n^L(t)-\tau\frac{f_n^L(t)}{I_n}}{R_n^U(t)\left[\frac{1}{\beta_n(t)}-\frac{a}{\beta_n(t)^k}\right]}$. Thus, the optimal solution of problem \eqref{time division1} can be given by
  	\begin{small}
  	\begin{align} \label{time division2}
  		&\left(\tau_n^U(t),\tau_n^D(t)\right)=
  		\nonumber\\ 
  		&\left\{
  		\begin{aligned}
  			&\left(0,0\right),\quad \quad \quad \quad \quad \quad \quad \quad \quad \quad \quad \quad \quad\quad \quad \quad  \text{if } \xi^U\geq0, \xi^D\geq0 ;\\
  			&\left(0,\min\left(\tau,\frac{Q_n^D(t)}{R_n^D(t)}\right)\right),\quad \quad \quad \quad \quad\quad \quad\quad\ \  \text{if } \xi^U\geq0, \xi^D<0;\\
  			&\left(\min\left(\tau,\theta\right),0\right),  \quad\quad\quad\quad\quad\quad\quad\quad\quad\quad\quad\quad \ \text{if } \xi^U<0, \xi^D\geq0;\\
  			&\left(\min\left(\tau,\theta\right),\tau-\min\left(\tau,\theta\right)\right),  \quad\quad\quad\quad\quad\quad\quad \text{if } \xi^U\leq\xi^D<0;\\
  			&\left(\tau-\min\left(\tau,\frac{Q_n^D(t)}{R_n^D(t)}\right),\min\left(\tau,\frac{Q_n^D(t)}{R_n^D(t)}\right)\right), \text{if } \xi^D\leq\xi^U<0.\\
  		\end{aligned}
  		\right.
  	\end{align}
  \end{small}
  	
  		Case 2: $\frac{1}{\beta_n(t)}-\frac{a}{\beta_n(t)^k}< 0$. In this case, $\tau_n^U(t)$ is lower bounded by $\theta$. Thus, the optimal solution to problem \eqref{time division1} can be given by
  		\begin{small}
  \begin{align}\label{time division3}
  		&\hspace{0em}\left(\tau_n^U(t),\tau_n^D(t)\right)=\hspace{0em}\nonumber\\
  		&\left\{
  		\begin{aligned}
  			&\left(\min(\tau,\max(0,\theta)),0\right),\quad \text{if } \xi^U\geq0, \xi^D\geq0 ;\\
  			&\left(\min(\tau,\max(0,\theta)),\min\left(\tau,\frac{Q_n^D(t)}{R_n^D(t)},\tau-\min(\tau,\max(0,\theta))\right)\right),\\
  			&\quad\quad\quad\quad\quad\quad\quad\quad\quad\quad \text{if } \xi^U\geq0, \xi^D<0 \text{ or } \xi^D\leq\xi^U<0;\\
  			&\left(\tau,0\right), \quad\quad\quad\quad\quad\quad\quad \text{if } \xi^U<0, \xi^D\geq0 \text{ or } \xi^U\leq\xi^D<0.
  		\end{aligned}
  		\right.
  	\end{align}
  	\end{small}

  	\subsection{Remote Computing Capacity Allocation Problem}
	According to problem \eqref{P4},  remote computing capacity allocation problem can be formulated by:
  	\begin{subequations}  \label{remote comp1}
  		\begin{align}
  			\min_{\pmb{f}^O(t)}\  & \sum_{n=1}^N\!\bigg[\!\!\left(\tilde{W}_n^4(t)H_n(t)\tau\!\!-\!\tilde{W}_n^3(t)\tau\right)\!\frac{f_n^O(t)}{G_n(t)I_n}\!+\!\!V\tau\kappa_M f_n^O(t)^3\!\bigg]\!,\\
  			\textrm{\textrm{s.t.}} \ 
  			&\tau \frac{f_n^O(t)}{G_n(t)I_n} \leq Q_n^O(t),\quad \forall n \in\mathcal{N},\\ &\sum_{n=1}^N f_n^O(t)\leq F_{MEC},
  		\end{align}
  	\end{subequations}
  	which is convex and the following corollary holds. 
  	\begin{coro} \label{coro3}
  		\emph{The optimal $f_n^O(t)$ is given by\begin{align}
  			f_n^O(t)\!=\!\left\{\begin{aligned}
  				&0,\quad\quad\quad\quad\quad\quad\quad\quad\quad\quad\quad\quad\   \text{if } \varphi_n(t)+\nu>0;\\
  				&\!\min\!\left\{\!\sqrt{\!-\!\frac{\varphi_n(t)\!+\!\nu}{3V\tau\kappa_M}},\frac{Q_n^O(t)G_n(t)I_n}{\tau}\right\}\!, \ \text{otherwise},
  			\end{aligned}
  			\right.
  		\end{align}
  	where $\nu\geq0$ is the  Lagrangian multiplier corresponding to (\ref{remote comp1}c)  and $\varphi_n(t)=\frac{\tilde{W}_n^4(t)H_n(t)\tau-\tilde{W}_n^3(t)\tau}{G_n(t)I_n}$. }
  	\end{coro}
  	\begin{proof}
  		Please refer to Appendix D.
  	\end{proof}

  \textit{Remark:} According to Corollary \ref{coro3}, the optimal $f_n^O(t)$ increases with $Q_n^O(t)$ since it is intuitive to allocate more computation capacity to keep remote processing queue stable. Moreover, $f_n^O(t)$ decreases as $Q_n^D(t)$ becomes large. This is because the output of remote processing queue is the input of downloading queue $Q_n^D(t)$. In order to keep $Q_n^D(t)$ stable, a smaller $f_n^O(t)$ is preferred.  Additionally, a larger $V$ results in a smaller $f_n^O(t)$ in order to reduce energy consumption at the cost of a longer queue length.

  	\vspace{-3ex}
  	\subsection{Performance Analysis}
  	\begin{algorithm}[t] 
  		\algsetup{linenosize= \footnotesize}   
  		\small
  		\caption{Lyapunov Optimization Based Online Optimization Algorithm}
  		\begin{algorithmic}[1]         \label{overall algorithm}
  			\STATE  {\bf Initialize:} $t=0$, $T$. 
  			\STATE {\bf While} $t\leq T$: 
  			\STATE  \hspace*{0.2in} Obtain system states $Q_n^L(t),Q_n^O(t),Q_n^D(t),X_n^q(t),X_n^r(t)$, $h_n(t),A_n(t),\pmb{\chi}_n(t)$. \\
  			\STATE\hspace*{0.2in} {\bf repeat} \\
  			\STATE\hspace*{0.4in} Solve local computing rate and transmit power problem \eqref{32}. \\
  			\STATE\hspace*{0.4in}  Solve extraction factor problem \eqref{Digestion factor Optimization}.\\
  			\STATE\hspace*{0.4in}  Solve time division problem \eqref{time division1}.\\
  			\STATE  \hspace*{0.2in} {\bf until} the objective value (\ref{P5}a) converges.\\
  			\STATE\hspace*{0.2in}  Solve remote computing capacity allocation problem \eqref{remote comp1}. \\
  			\STATE \hspace*{0.2in}Update $Q_n^L(t),Q_n^O(t),Q_n^D(t),X_n^q(t),X_n^r(t),\pmb{\chi}_n(t)$ and set $t=t+1$.
  			\STATE {\bf End}
  		\end{algorithmic}
  	\end{algorithm}
  \begin{figure}[t]
  	\centering
  	\includegraphics[width=0.5\textwidth]{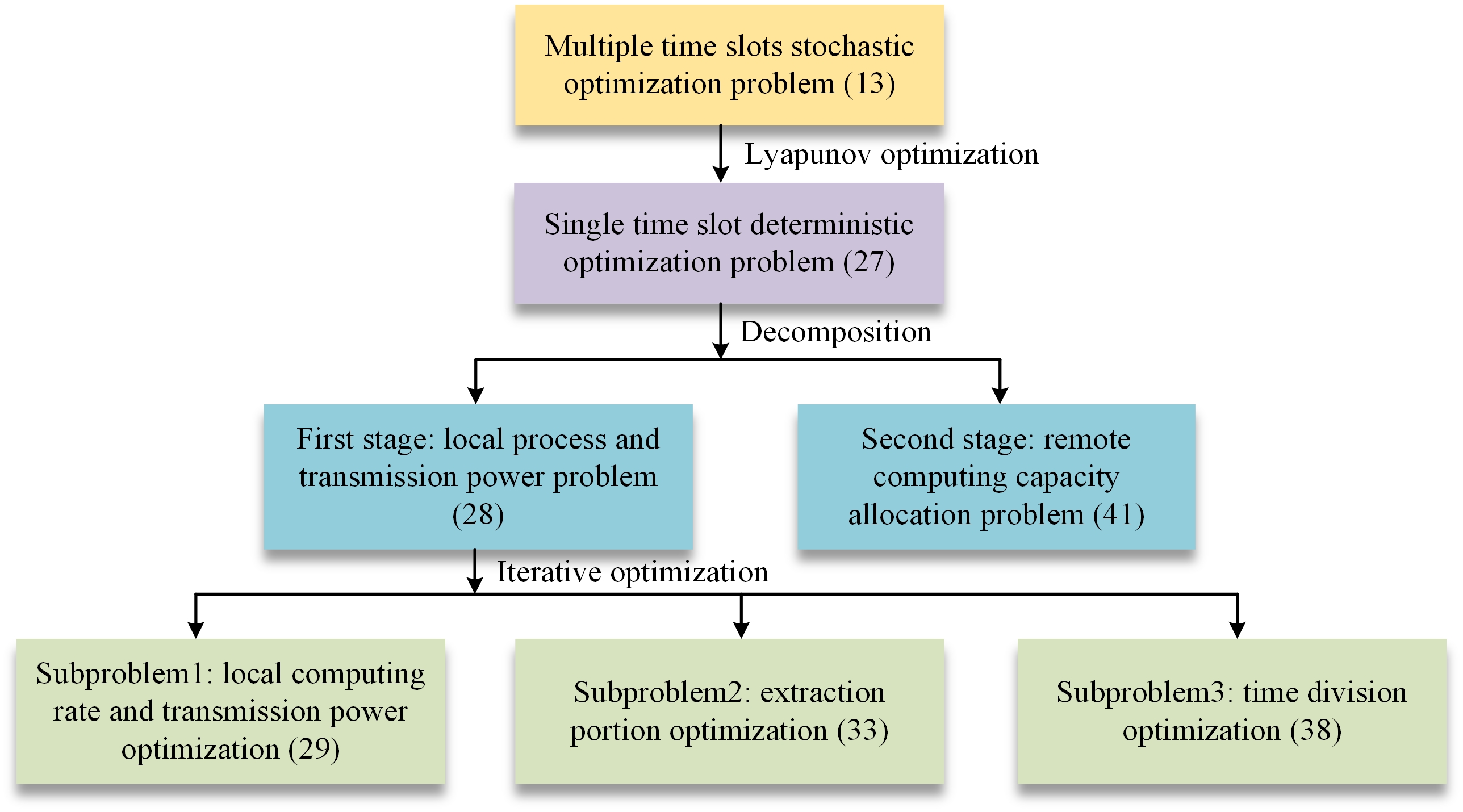}
  	\caption{\textcolor{black}{The overall procedure  of the proposed algorithm.}} \label{algorithm flow fig}
  	\vspace{-1.5em}
  \end{figure} 
  	\textcolor{black}{The overall Lyapunov optimization based online optimization algorithm is summarized in Algorithm~\ref{overall algorithm} and the overall algorithm procedure  of is depicted in Fig.~\ref{algorithm flow fig}. 
  	At each time slot, deterministic problem \eqref{P4} is solved once.  The complexity of solving deterministic problem \eqref{P4} lies in solving four subproblems: 
  	1) For local computing rate and transmit power optimization,   the bisection method for searching numbers of downlink transmit power optimization, as well as local computing rate and uploading power optimization can be  estimated as $\mathcal{O}\left(\log_2 \left(1/\epsilon_1\right)\right)$ and $\mathcal{O}\left(\log_2 \left(1/\epsilon_2\right)\right)$, respectively. Therefore, according to Corollary \ref{coro1}, the complexity of this subproblem is $\mathcal{O}\left(2N\log_2 \left(1/\epsilon_1\right)+N\log_2 \left(1/\epsilon_2\right)\right)$. 2) For extraction factor optimization, the iteration number of SCA algorithm can be estimated as $\mathcal{O}\left(1/\sqrt{\epsilon_5}\right)$, where $\epsilon_5$ is the accuracy of SCA algorithm \cite{cot}. Due to the fact that at each iteration, the bisection method will be adopted once, the complexity of solving this problem can be given by $\mathcal{O}\left(N/\sqrt{\epsilon_5}\log_2 \left(1/\epsilon_3\right)\right)$, where $\epsilon_3$ is the accuracy of dual variable $\xi$. 3) For time division optimization, the complexity is given by $\mathcal{O}(N)$ according to \eqref{time division2} and \eqref{time division3}. 4) For remote computing capacity allocation problem, the complexity is $\mathcal{O}\left(2N\log_2\left(1/\epsilon_4\right)\right)$, where $\epsilon_4$ is the precision of dual variable $\nu$. In summary, the proposed algorithm can be implemented with linear complexity. 
  	Next, we analyze the convergence of the proposed BCD algorithm. The objective of \eqref{P5} is guaranteed to be nonincreasing solving each of the above subproblems. Moreover, there exists a lower bound of the  objective.  Therefore, the proposed algorithm is guaranteed to converge to a sub-optimal solution. 
  	As can be seen, the proposed algorithm specifies the solution of a stochastic optimization problem measured from a long-term perspective to each time shot. In addition, the algorithm does not require any  information from future system states, nor does it needs to know the probability distributions of random events. }
  	\textcolor{black}{\begin{lemma} \label{lemma 2}
  		\emph{Denote $E^{opt}$ as the minimum objective (\ref{P1}a) over the feasible region.  Assuming that problem \eqref{P1} is feasible and there exists a certain $\psi > 0$ and $\Phi(\psi)>E^{opt}$, we have the following properties:
  		\begin{itemize}
  			\item All the virtual queues are mean rate stable and time average constraints  are satisfied with probability 1. 
  			\item The upper bound of the time average system energy consumption under the proposed algorithm is given by
  			\begin{align}
  				\lim_{T\rightarrow\infty}\frac{1}{T}\sum_{t=1}^T\sum_{n=1}^N\mathbb{E}[E_n(t)]\leq E^{opt} +\frac{\Xi+\Gamma}{V}, 
  			\end{align}
  			where $\Gamma$ is the upper bound of the difference between the RHS of \eqref{theorem_1_eq} and its $\Gamma$-approximation (\ref{P4}a) over the feasible region.   
  			\item The time average sum lengths of all the virtual queues is upper bounded by
  			\begin{align}
  				\lim_{T\rightarrow\infty}\frac{1}{T}\sum_{t=1}^T\sum_{n=1}^N\mathbb{E}\left[X_n^q(t)+X_n^r(t)\right]\nonumber\\
  				\leq \frac{\Xi+\Gamma+V\left(\Phi(\psi)-E^{opt}\right)}{\psi}.
  			\end{align}
  		\end{itemize} }
  	\end{lemma}}
  	
  	\begin{proof}
  		This lemma comes as a direct consequence of  \cite[Theorem 4.8]{Neely2010Stochastic}.
  	\end{proof}
  	
  	\textcolor{black}{Lemma~\ref{lemma 2} claims that the proposed algorithm can guarantee the time average constraints (\ref{P4}b) and (\ref{P4}c). It also indicates that there exists a limited optimality gap when solving problem \eqref{P4} rather than  problem \eqref{P1} directly. With control parameter $V$ rising, the system long-term energy consumption reduces inversely proportional to $V$. Meanwhile, the backlog of virtual increases with a speed of $V$.  In other words, there exists an $\left[\mathcal{O}(1/V),\mathcal{O}(V)\right]$ tradeoff between energy consumption and delay, which will be demonstrated in simulations. Therefore, the performance of the proposed algorithm can be analytically characterized which greatly facilitates the implementation of the required energy consumption and delay tradeoff in practical dynamic semantic-aware MEC systems. }

\section{Simulations}
	\begin{table}[t]
	\caption{Simulation parameters} 
	\label{table} 
	\centering 
	\begin{tabular}{|m{2.7cm}<{\centering}|m{2.3cm}<{\centering}|m{2.4cm}<{\centering}|}
		\hline
		\hline
		$\sigma^2=-174$ dBm/Hz&$B=2$ MHz&$I_n=70$ cycles/bit\\
		\hline
		$\kappa_n=10^{-26}$ & $\kappa_M=10^{-26}$ &$\tau=1$ second\\
		\hline
		$V=10^{16}$&$a=10^{-3}$&$k=4$\\
		\hline
		$Q_n^{avg}=20$ Mbps & $R_n^{avg}=4$ Mbits&$U=1/100$\\
		\hline
		$f_n^{\max}=1$ GHz&$F_{MEC}=30$ GHz&$p_n^{\max}=0.3$ Watt\\
		\hline
		$P_{MEC}=1.5$ Watts& $\beta_n^{\min}=0.3$& $p=1$\\
		\hline
		\hline
	\end{tabular}\vspace{-2em}
\end{table}

\begin{figure*} 
	\centering 
	\subfigure[Energy consumption $E$ versus $t$]{\label{conv_fig:a}
		\includegraphics[width=0.32\linewidth]{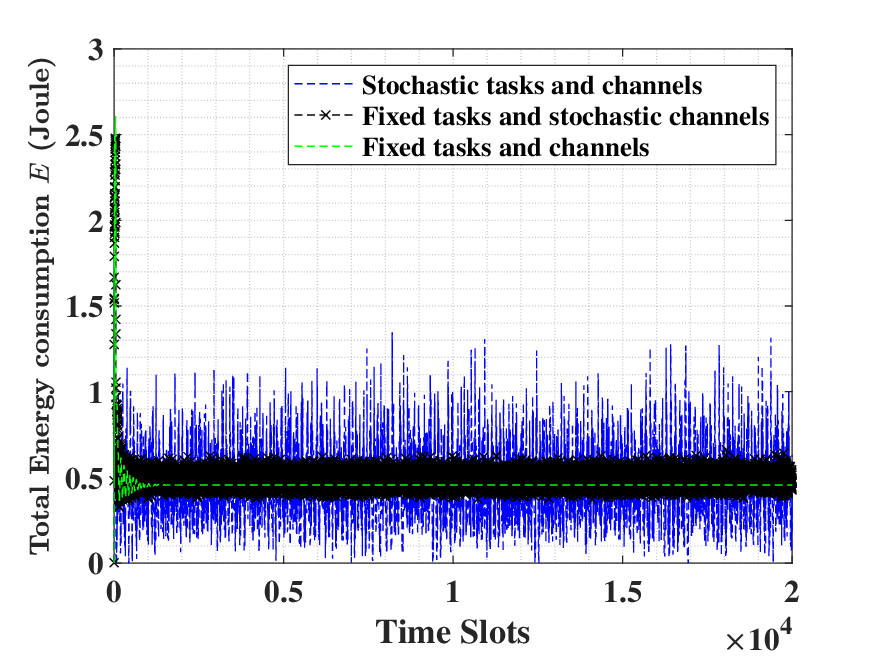}}
	\hspace{-0.01\linewidth}
	\subfigure[Average data queue length $Q^{total}$ versus $t$]{\label{conv_fig:b}
		\includegraphics[width=0.32\linewidth]{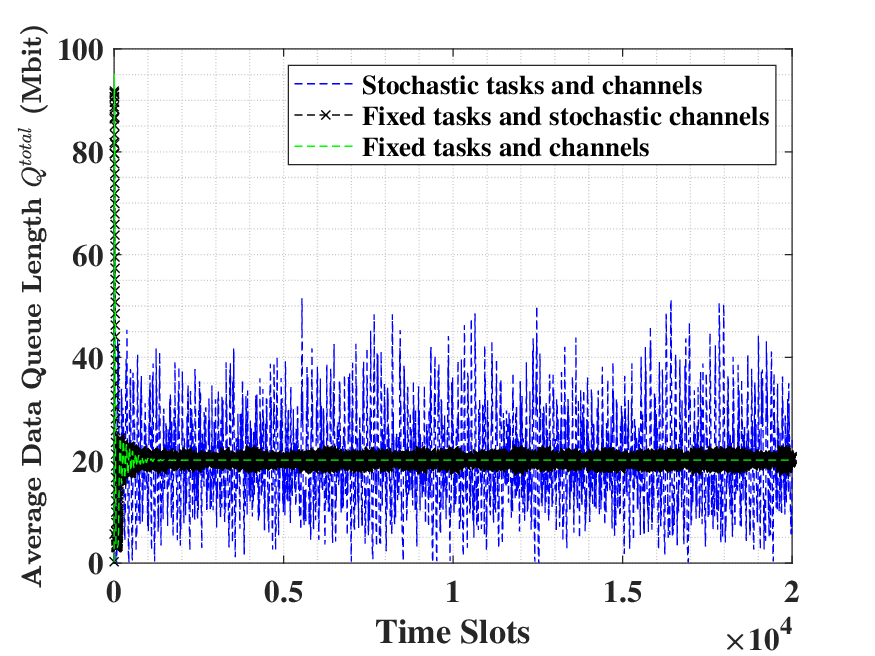}}
	\hspace{-0.01\linewidth}
	\subfigure[Average rate queue length $R$ versus $t$]{\label{conv_fig:c}
		\includegraphics[width=0.32\linewidth]{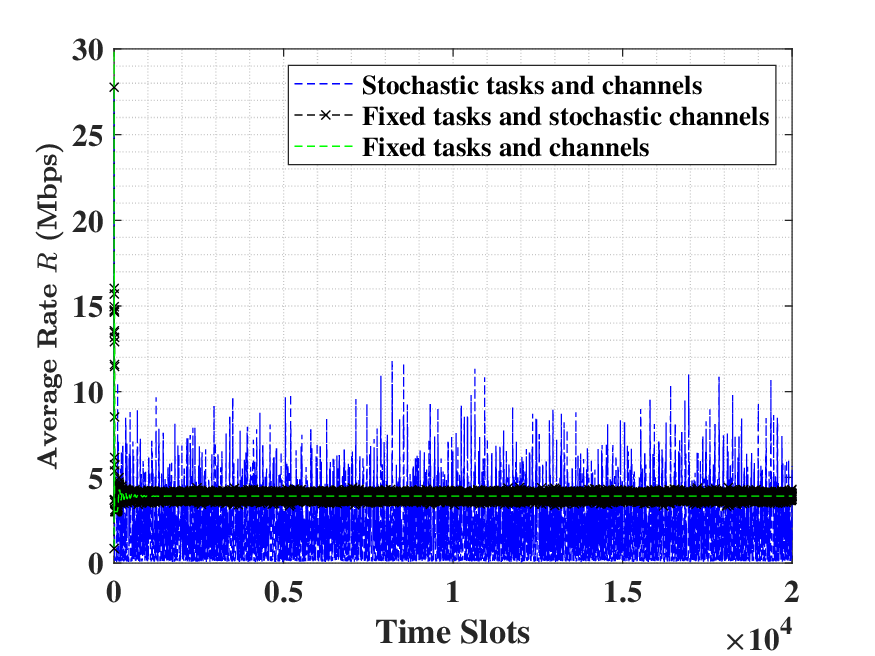}}
	\caption{Convergence performance under  dynamic and static scenarios.}
	\label{conv_fig}
\end{figure*}

In this section, numerical simulations are conducted to evaluate the 
performance of the proposed algorithm.  There are  $N=10$ TDs and the distances between TDs and BS are evenly set in $[120,255]$ meters. \textcolor{black}{In channel model, we set antenna gain $A=3$, carrier frequency $f_c=915$ MHz, path-loss factor $\ell=3$, speed of light $c=3\times10^8$ m/s, and the LoS link factor of Rician fading is $\gamma=0.3$.} A kind of object recognition task arriving at each time slot obeys exponential distribution with mean value $\mathbb{E}[A_n(t)]=\lambda_n=3$ Mbits.  
The other parameters are set as in Table~\ref{table} unless otherwise mentioned. All the simulations run over twenty thousand time  slots. \textcolor{black}{The following benchmark algorithms are provided to contrast with the proposed DRMSA algorithm: 1) No Semantic-aware Algorithm (NS)\cite{9449944}: the algorithm without semantic aware technology, i.e., $\beta_n(t)=1$ for all TDs and time slots. This algorithm corresponds to the conventional dynamic MEC scenario that TDs directly upload raw data to the MEC server. 2) No Local Computing Algorithm (NL)\cite{9242286}: the algorithm without local computing units. 
3) Myopic Method \cite{8771176}: the greedy algorithm ignores the effects of queues states and minimizes energy consumption under the processing rate constraint at each time slot. 
4) Exhausting Search Algorithm (EXH): the algorithm randomly sets $50$ initial points in Algorithm~\ref{overall algorithm} at each time slot in order to obtain a near global optimal solution. }     

To evaluate the performance of the DRMSA algorithm under different environments,  Fig.~\ref{conv_fig} illustrates the convergence performance under three kinds of scenarios: stochastic tasks and channels; fixed tasks and stochastic channels; fixed tasks and channels. Specifically, stochastic tasks mean that the sizes of tasks follow an exponential distribution with mean value $\lambda_n=3$ Mbits, while the fixed tasks indicate that the  tasks at each time slot are exactly $3$ Mbits.  Stochastic channels mean that the channel states vary across different time slots while channels are stable in fixed channels scenarios. According to Fig.~\ref{conv_fig}, this algorithm converges in less than $2000$ time slots under all the previously mentioned environments.  The energy consumption, as well as queue states fluctuate sharply in the beginning. Afterward, the fluctuations are stable. This indicates the DRMSA algorithm can well adapt to dynamic scenarios. We can also find that all of these three scenarios achieve almost the same performance, i.e., an average of $0.45$ Joule/s in energy consumption, $20$ Mbits in data queue length, as well as $4$ Mbits in rate queue length, while the fluctuation of stochastic tasks and channels case is the largest, and that of the fixed tasks and channels case is the smallest.   This implies that the proposed algorithm can effectively adjust the resource allocation at each time slot by observing current states including arriving tasks, channel states, and queues backlogs to maintain energy efficiency.

\begin{figure*} 
	\centering 
	\subfigure[Energy consumption $E$ versus $V$]{\label{V_fig:a}
		\includegraphics[width=0.32\linewidth]{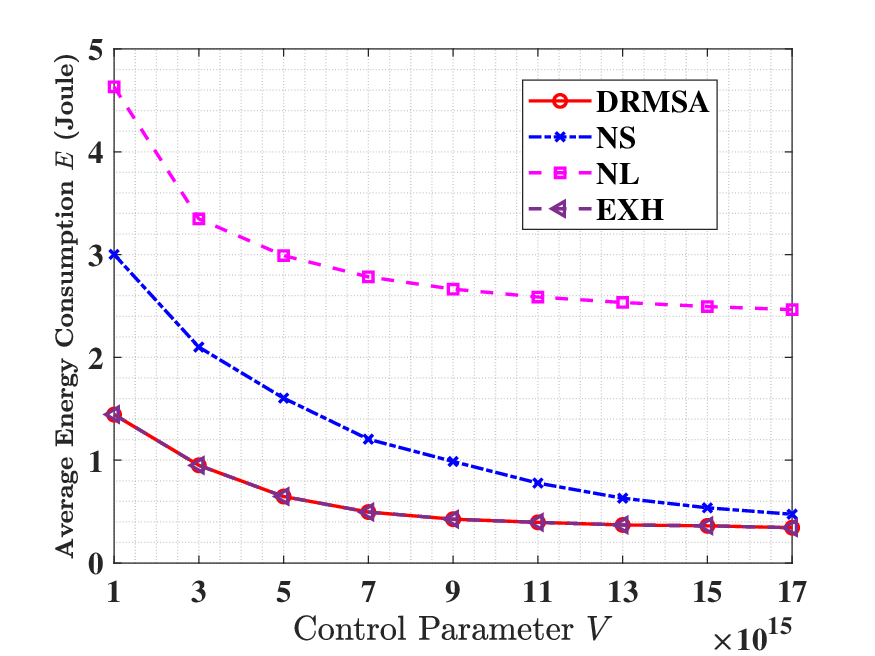}}
	\hspace{-0.01\linewidth}
	\subfigure[Average actual data queue length $Q^{total}$ versus $V$]{\label{}
		\includegraphics[width=0.32\linewidth]{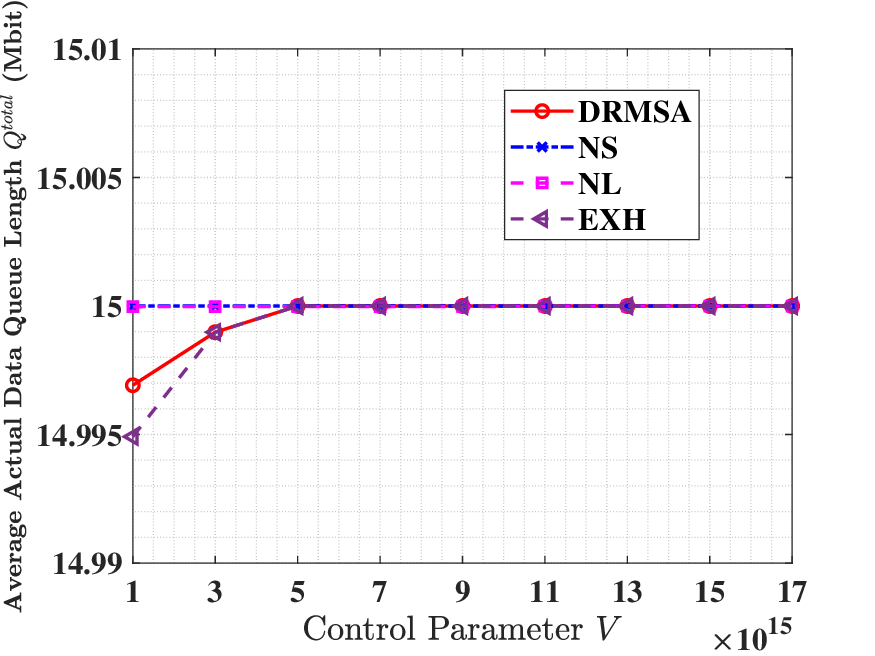}}
	\hspace{-0.01\linewidth}
	\subfigure[Average rate $R$ versus $V$]{\label{}
		\includegraphics[width=0.32\linewidth]{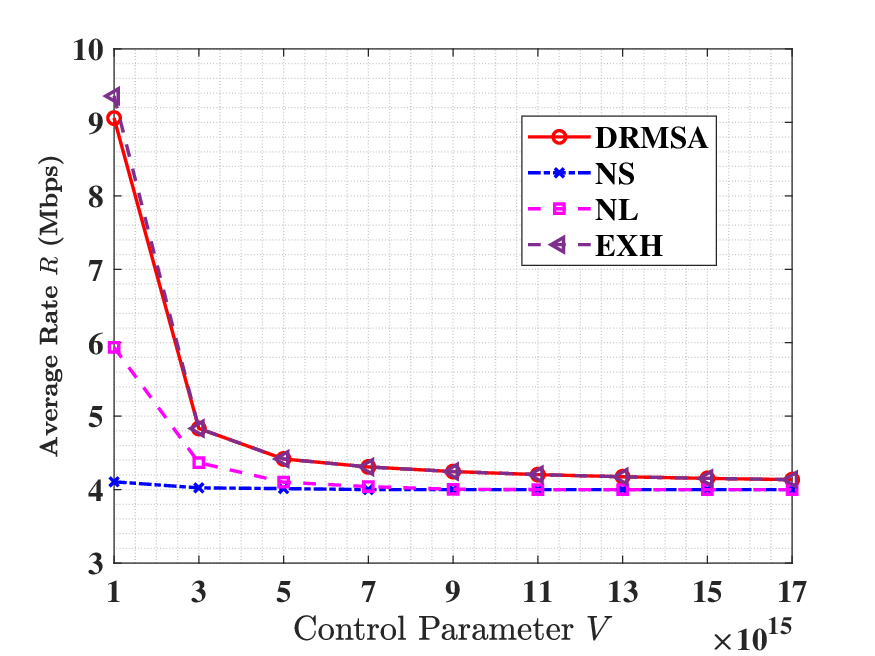}}
	
	
	\caption{\textcolor{black}{Performance comparisons between different algorithms under different control parameters $V$.}}
	\label{V_fig}
\end{figure*}

\textcolor{black}{Fig.~\ref{V_fig} shows the performance comparisons between different algorithms under different control parameters $V$. In Fig.~\ref{V_fig}(a), the energy consumption of all algorithms decreases with $V$. The DRMSA algorithm achieves almost the same performance as the EXH algorithm. Moreover,  the DRMSA algorithm demonstrates its  superiority over the other two schemes. 
This can be explained by that the NS algorithm has to upload the whole raw data that is much larger than extracted semantic information, leading to larger transmit power. For NL algorithm, since it cannot compute locally, the backlogs of queues are  more severe than those of DRMSA algorithm. High backlogs enforce additional energy to maintain queues stable. 
From Fig.~\ref{V_fig}(b) and Fig.~\ref{V_fig}(c), we can find that the actual data queue length is strictly smaller  than the  predefined $Q^{avg}$ and the rate is larger than $R^{avg}$, which demonstrates the efficiency of the proposed algorithm.  Furthermore, the DRMSA and EXH algorithm achieve the largest margins among all these algorithms in terms of actual data and rate queues length. This is due to the fact that the defects of NS and NL algorithms in processing rate, they have to enlarge queues backlogs in order to save energy. Last but not least, as the control parameter $V$ increases, the energy consumption reduces and virtual queue length becomes large. This proves the theoretical analysis of a  $[\mathcal{O}(1/V),\mathcal{O}(V)]$ tradeoff between system energy consumption and queues backlogs.}
  
\begin{figure} 
	\centering 
	\subfigure[Average Energy consumption $E/N$ versus $\lambda_n$]{\label{}
		\includegraphics[width=0.9\linewidth]{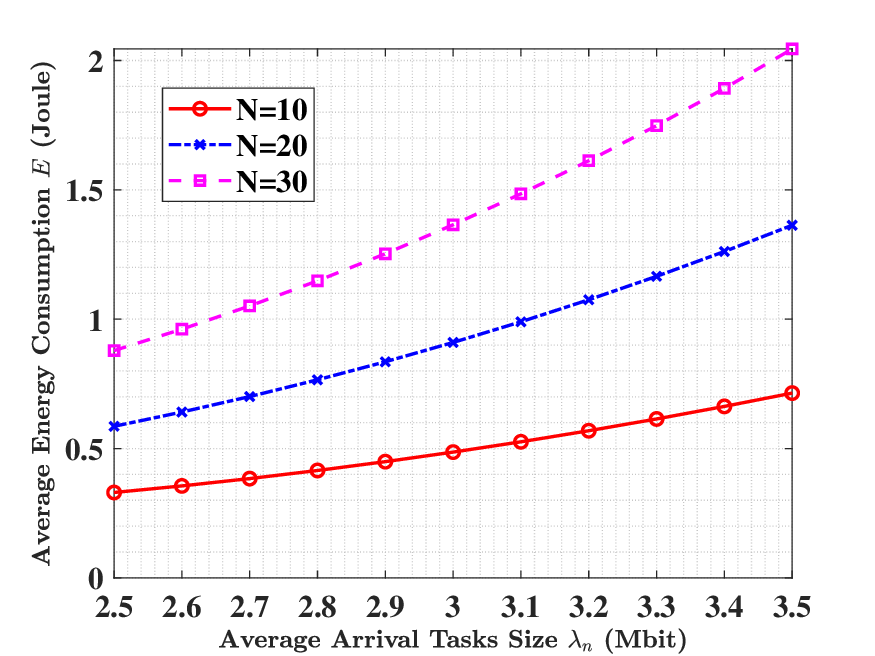}}
	\hspace{-0.01\linewidth}
	\hspace{-0.01\linewidth}
	\subfigure[Average rate $R$ versus $\lambda_n$]{\label{}
		\includegraphics[width=0.9\linewidth]{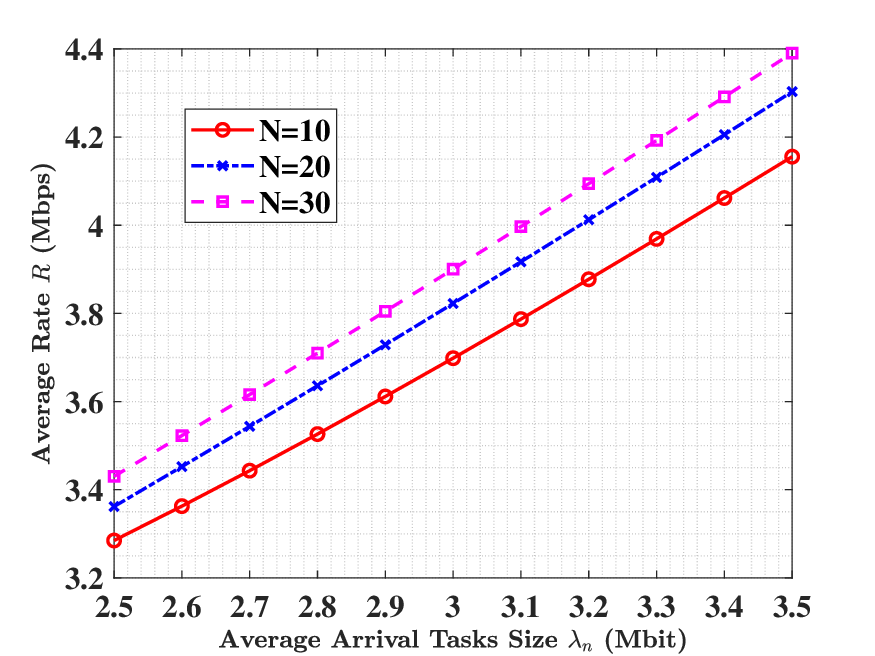}}
	
	
	\caption{\textcolor{black}{Performance versus $\lambda_n$ with different $N$.}}
	\label{A_fig}
\end{figure}
  
  \textcolor{black}{In Fig.~\ref{A_fig}, we present the performance comparisons under different tasks arrival sizes $\lambda_n$ and number of TDs $N$. As can be seen in Fig.~\ref{A_fig}(a), the average energy consumption increases with larger  arrival tasks size and more TDs due to the heavier network load. In Fig.~\ref{A_fig}(b), the average rate monotonously increases with tasks size. This is because as the size of tasks (i.e., the input of the queue) increases, the processing rate (i.e., the output of the queue) should be improved correspondingly in order to maintain data queue stability.} 
  
  \begin{figure} [t]
  	\centering 
  	
  		\includegraphics[width=0.9\linewidth]{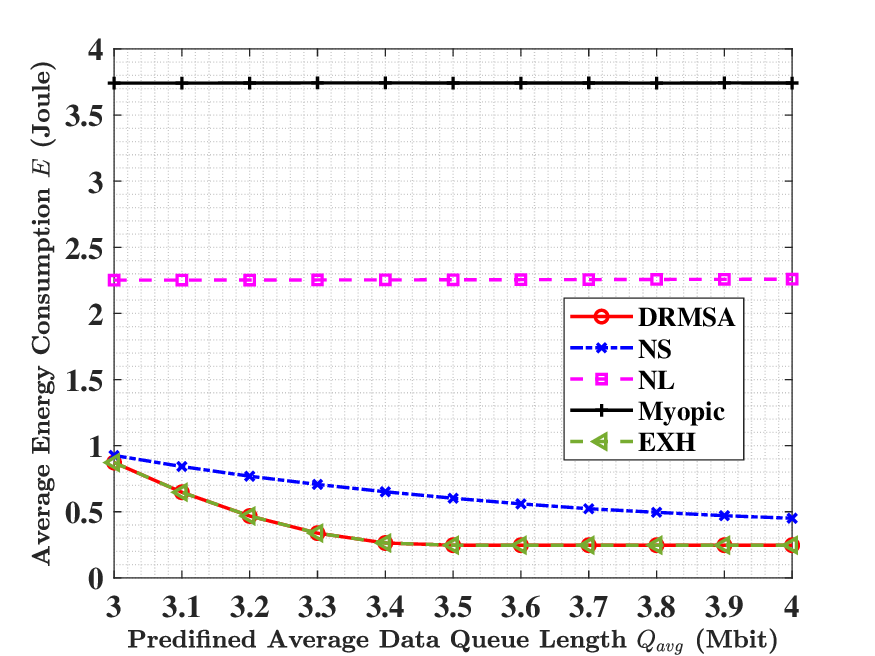}
  	
  	\caption{\textcolor{black}{Performance comparisons between different algorithms under different predefined data queue length $Q^{avg}$.}}
  	\label{Qavg_fig}
  \end{figure}
  
  \textcolor{black}{Performance comparisons between different algorithms with predefined data queue lengths vary from $3$ to $4$ Mbits are demonstrated in Fig.~\ref{Qavg_fig}. In the considered predefined data queue length region, the Myopic algorithm is unable to  stabilize data queues under the stochastic channels and tasks scenarios. Besides, it consumes the largest energy compared with other schemes which implies the necessity of long-term optimization in stochastic environment.  
  The energy consumption of the NL algorithm is larger than those of the DRMSA and EXH algorithms due to its weak processing rate. 
  We can also observe that as the average data queue length gets larger, the DRMSA algorithm becomes more energy efficient, implying that it can flexibly adjust semantic extraction proportions according to the data queue length requirements. 
 According to equation \eqref{Q_n^L(t+1)},  semantic extraction brings additional workloads which are reflected in the length of local processing queue. Meanwhile, the smaller extraction factor leads to more workloads as well as longer data queues. Therefore, as the predefined data queue length increases, the DRMSA algorithm has the opportunity to search for a lower extraction factor, which results in a higher processing rate, so as to reduce the energy consumption of the systems as shown in Fig.~\ref{Qavg_fig}. This phenomenon shows the energy efficiency of semantic-aware technology in MEC systems. Compared with NS, NL and Myopic algorithms, the proposed algorithm yields up to $41.8\%$, $83.58\%$ and $90.11\%$ energy reduction, respectively. }
  
  \begin{figure}
  	\centering 
  	\includegraphics[width=0.9\linewidth]{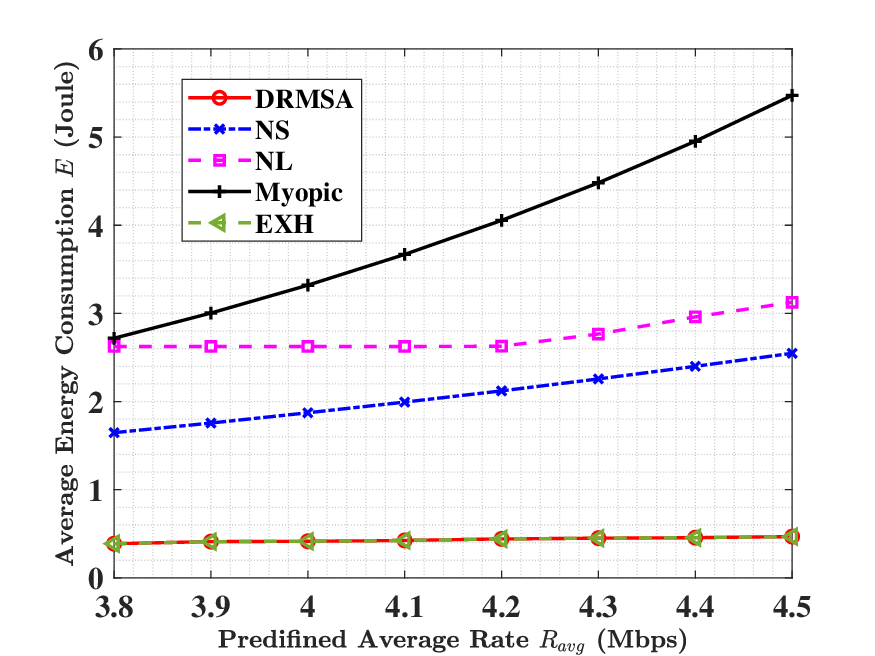}
  	\caption{\textcolor{black}{Performance comparisons between different algorithms under different predefined average rate $R^{avg}$.}} 
  	\label{Ravg_fig}
  \end{figure}
  
  \textcolor{black}{Fig.~\ref{Ravg_fig} presents performance of different algorithms with predefined average rates vary from $3.8$ to $4.5$ Mbps. All algorithms except Myopic algorithm satisfy the long-term data queue length and average rate constraints, while DRMSA and EXH algorithms achieves the lowest energy consumption. 
  The energy consumption of NL algorithm stays unchangeable at the beginning and then rises as the predefined rate increases. This can be explained by that 
  when the predefined rate is small,  NL algorithm has to promote its processing rate in order to stabilize data queues due to its weak processing capacity, while When the predefined rate is higher, the predefined rate constraints urge the system to speed up. 
  For the NS algorithm, since it has no semantic extracting processing unit, its processing rate is weak compared with DRMSA. Hence, the predefined rate constraints always work, which explains the monotonicity of NS algorithm in terms of energy consumption. Compared with other schemes, the energy consumption of DRMSA and EXH algorithms are less sensitive to the variations of the predefined average rate due to their flexible semantic extraction capacities. This shows the robustness of DRMSA algorithm in terms of high  workload.  }    
  
  \begin{figure*} 
  	\centering 
  	\subfigure[Average energy consumption $E/N$ versus $\beta^{\min}$]{\label{}
  		\includegraphics[width=0.32\linewidth]{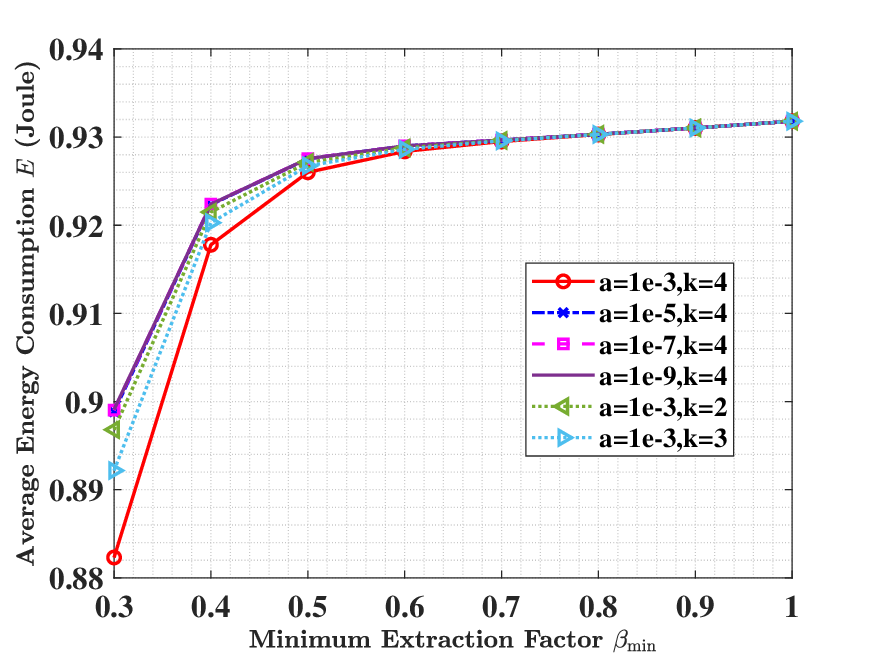}}
  	\hspace{-0.01\linewidth}
  	\subfigure[Average actual data queue length $Q^{total}$ versus $\beta^{\min}$]{\label{}
  		\includegraphics[width=0.32\linewidth]{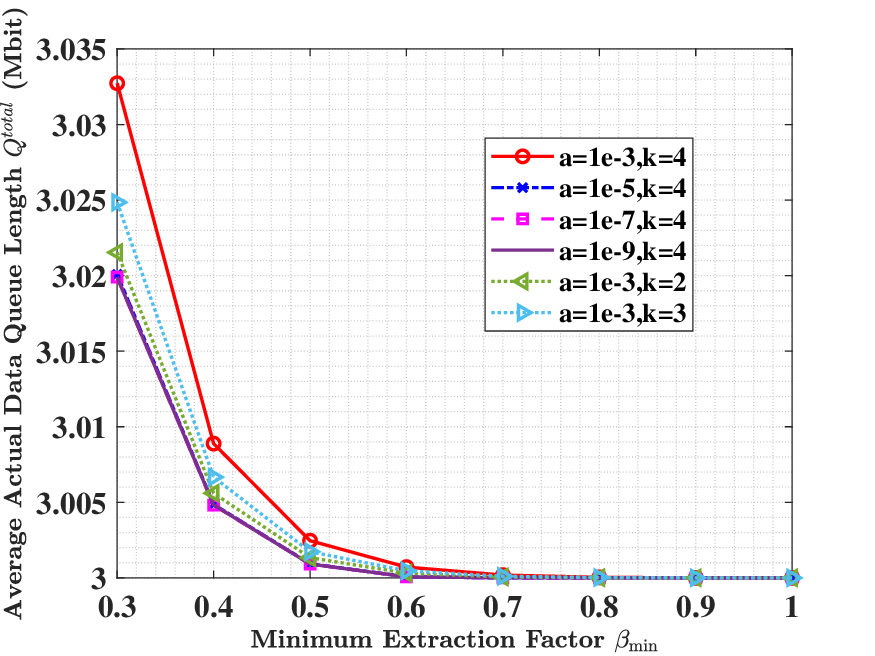}}
  	\hspace{-0.01\linewidth}
  	\subfigure[Average rate $R$ versus $\beta^{\min}$]{\label{}
  		\includegraphics[width=0.32\linewidth]{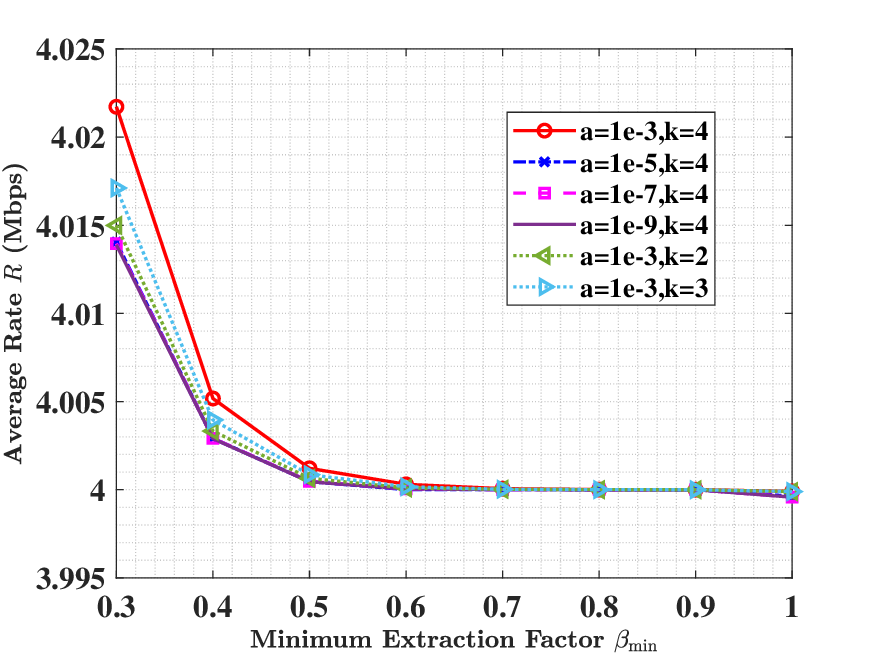}}
  	
  	
  	\caption{\textcolor{black}{Performance comparisons between different workload expressions $C_n(t)$ under different minimum extraction factors $\beta^{\min}$.}}
  	\label{beta_fig}
  \end{figure*}
  
  \textcolor{black}{To test the algorithm in applications to various semantic-aware scenarios, we plot the performance comparisons between different workload expressions under different minimum extraction factors in Fig.~\ref{beta_fig}. Note that the case when $\beta^{\min}=1$ is equivalent to NS algorithm. As can be seen, as the minimum extraction factor decreases from $1.0$ to $0.3$, the energy consumption reduces, data queue length, which can also be regarded as delay, and the average rate increases. This phenomenon reflects a tradeoff between delay and energy consumption in terms of extraction factor. This is due to the fact that a smaller extraction factor significantly reduces the amount of data to be uploaded, thus the transmit power is reduced and the  system average rate is promoted. Meanwhile, a smaller extraction factor means more workloads to implement more purified semantic extractions that brings additional time consumption.  Furthermore, the performance of the proposed algorithm is alike under different workload expressions. This shows its ability to apply to different semantic-aware scenarios.} 
  
%
%
%
%

    \begin{figure*} 
  	\centering 
  	\subfigure[Average energy consumption $E/N$ versus $\beta^{\min}$]{\label{}
  		\includegraphics[width=0.32\linewidth]{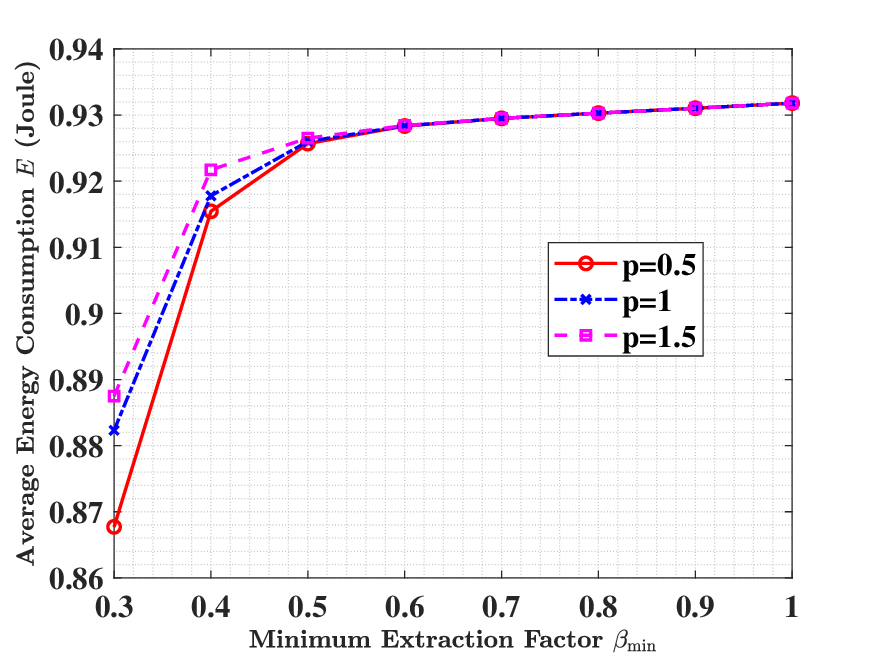}}
  	\hspace{-0.01\linewidth}
  	\subfigure[Average actual data queue length $Q^{total}$ versus $\beta^{\min}$]{\label{}
  		\includegraphics[width=0.32\linewidth]{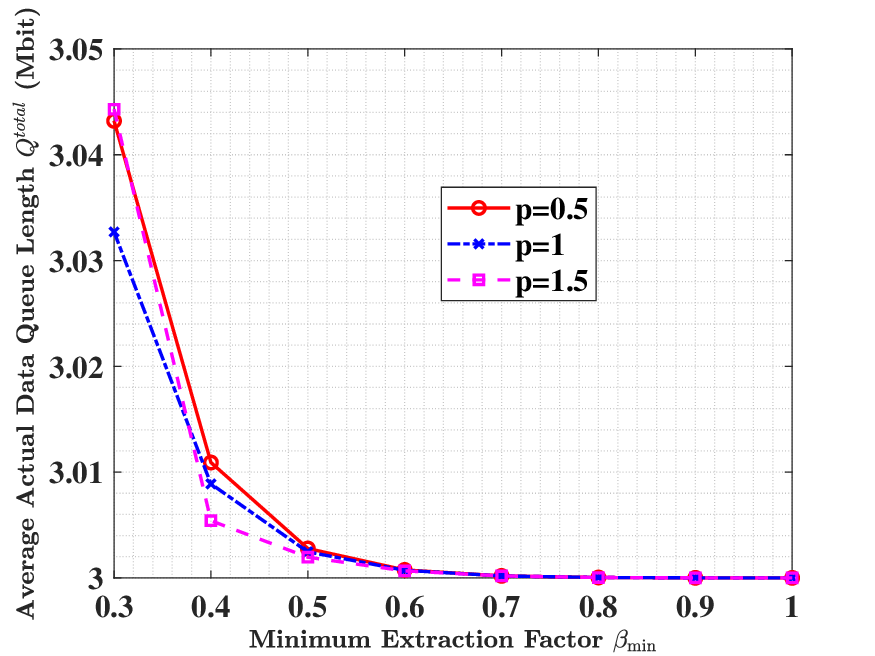}}
  	\hspace{-0.01\linewidth}
  	\subfigure[Average rate $R$ versus $\beta^{\min}$]{\label{}
  		\includegraphics[width=0.32\linewidth]{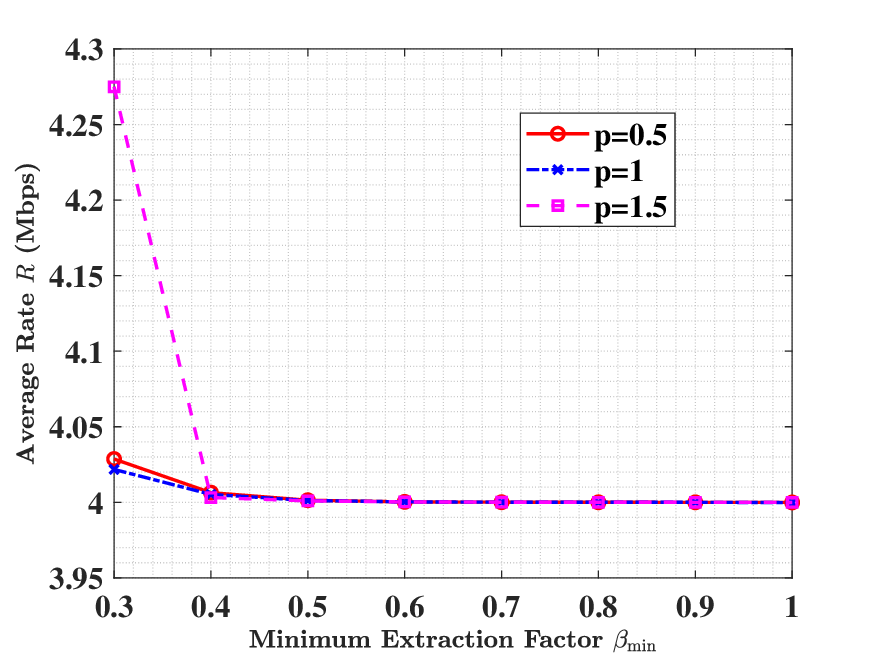}}
  	
  	
  	\vspace{-1em}
  	\caption{\textcolor{black}{Performance comparisons between different $G_n(t)$ under different minimum extraction factors.}}
  	\label{betaq_fig}
  \end{figure*}

\textcolor{black}{In Fig.~\ref{betaq_fig}, the performance of the proposed algorithm under different $G_n(t)=\frac{1}{\left(\min\{\pmb{\chi}_n(t)\}\right)^p}$ is depicted. The overall performance trends with respect to $\beta_{\min}$ are similar to those of Fig.~\ref{beta_fig}, which shows the proposed algorithm applies to the generality of $G_n(t)$. Moreover, a larger $p$ brings more energy consumption. This is due to the fact that $G_n(t)$ becomes larger when $p$ gets bigger, which indicates that processing a bit of semantic data requires more CPU cycles. Thus more remote computation resource is needed to maintain system queues stable resulting in more energy consumption.   }

  \section{Conclusion}
  In this paper, we have proposed a  joint communication and computation resource allocation framework for long-term MEC systems, where stochastic tasks arrive at TDs and channels vary with time, with the aid of prevalent semantic transmission technology. An optimization problem has been formulated to dynamically optimize semantic extraction factors, communication and computation resources to achieve the minimum long-term  energy consumption. A Lyapunov optimization-based low-complexity online algorithm is proposed to solve this stochastic optimization problem. Simulation results demonstrate the superior performance of
  the proposed algorithm over the benchmark schemes with respect to energy consumption. Furthermore, the tradeoff between energy consumption and delay is well balanced by the proposed algorithm. We also observe that energy consumption of the whole system is significantly reduced as the minimum semantic extraction factor gets small. 
  \textcolor{black}{For future works, we will further investigate semantic-aware MEC scenarios with more advanced communication techniques such as multiple-input multiple-output (MIMO), non-orthogonal multiple access (NOMA), and near-field transmission.   Moreover, we will extend our work to other scenarios such as satellite communications. We also expect to set up a hardware semantic-aware MEC platform with specific tasks in the future.}

  \vspace{-1.5em}
  \appendices
  \section{Proof of Theorem \ref{theorem_1}}  \label{proof_1}
  \setcounter{equation}{0}
  \renewcommand\theequation{A.\arabic{equation}}
 First, for virtual queue $X_n^q(t)$, we have 
 \begin{small}
 \begin{align}  \label{A.1}
 	&\frac{1}{2}X_n^q(t+1)^2-\frac{1}{2}X_n^q(t)^2\nonumber\\
 	 &\overset{(a)}{\leq} \frac{1}{2}\left(X_n^q(t)+Q_n^{total}(t+1)-Q_n^{avg}\right)^2-\frac{1}{2}X_n^q(t)^2\nonumber,\\
 	&\overset{(b)}{\leq} \frac{1}{2}Q_n^{total}(t+1)^2+\frac{1}{2}(Q_n^{avg})^2+X_n^q(t)(Q_n^{total}(t+1)-Q_n^{avg}), 
 \end{align}
\end{small}where inequality $(a)$ holds because of $(\max\{0,x\})^2\leq x^2$; $(b)$ is due to the deletion of non-positive term $-Q_n^{total}(t+1)Q_n^{avg}$. Further, we have
\begin{small}
 \begin{align}  \label{A.2}
 	&Q_n^{total}(t+1)^2=\left(Q_n^L(t+1)+\left(Q_n^O(t+1)+Q_n^D(t+1)\right)\right)^2\nonumber\\
 	&\overset{(c)}{\leq} 2Q_n^L(t+1)^2+2\left(Q_n^O(t+1)+Q_n^D(t+1)\right)^2\nonumber,\\
 	&\overset{(d)}{\leq} 2Q_n^L(t+1)^2+ 4Q_n^O(t+1)^2 +4Q_n^D(t+1)^2,
 \end{align}
\end{small}where inequalities $(c)$ and $(d)$ hold because of 	$(x+y)^2\leq2x^2+2y^2$. 
Through plugging \eqref{Q_n^L(t+1)}, \eqref{Q_n^O(t+1)} and \eqref{Q_n^D(t+1)} into \eqref{A.2}, we have
\begin{small}
 \begin{align}  \label{A.3}
 	&\frac{1}{2}Q_n^{total}(t+1)^2\overset{(e)}{\leq} \left(Q_n^L(t)+C_n(t)\right)^2 +A_n(t)^2+2Q_n^O(t)^2\nonumber\\
 	&+ \left(\tau R_n^L(t)+\frac{\tau_n^U(t)R_n^U(t)}{\beta_n(t)}\right)^2+4Q_n^O(t)\left(\tau_n^U(t)R_n^U(t)-\tau R_n^M(t)\right)\nonumber\\
 	&+2\left(Q_n^L(t)+C_n(t)\right)\left(A_n(t)-\tau R_n^L(t)-\frac{\tau_n^U(t)R_n^U(t)}{\beta_n(t)}\right)  \nonumber\\
 	&+ 2\tau^2R_n^M(t)^2 +  2\tau_n^U(t)^2R_n^U(t)^2+2Q_n^D(t)^2+2\tau_n^D(t)^2R_n^D(t)^2\nonumber\\
 	&+ 2H_n(t)^2\tau^2R_n^M(t)^2+4Q_n^D(t)\left(H_n(t)\tau R_n^M(t)-\tau_n^D(t)R_n^D(t)\right)\nonumber,\\
 	&\overset{(f)}{\leq} Q_n^L(t)^2+(C_n^{\max})^2+2Q_n^L(t)C_n(t)+A_n(t)^2+2Q_n^O(t)^2\nonumber\\
 	&+\tau^2\left(R_n^{L,\max}+\frac{R_n^{U,\max}}{\beta_n^{\min}}\right)^2+4Q_n^O(t)\left(\tau_n^U(t)R_n^U(t)-\tau R_n^M(t)\right)\nonumber\\
 	&+2\left(Q_n^L(t)+C_n^{\max}\right)\left(A_n(t)-\tau R_n^L(t)-\frac{\tau_n^U(t)R_n^U(t)}{\beta_n(t)}\right)\nonumber\\
 	&+2\tau^2(R_n^{M,\max})^2+2\tau^2(R_n^{U,\max})^2+2Q_n^D(t)^2+2\tau^2(R_n^{D,\max})^2\nonumber\\
 	&+2(H_n(t)\tau R_n^{M,\max})^2+4Q_n^D(t)\left(H_n(t)\tau R_n^M(t)-\tau_n^D(t)R_n^D(t)\right),
 \end{align}
\end{small}where $C_n^{\max}$ is the maximum workload for semantic extraction, $R_n^{L,\max}$ is the maximum local processing rate in \eqref{R_n^L(t)}, $R_n^{U,\max}$ is the maximum uplink transmission rate in \eqref{R_n^U(t)}, $R_n^{M,\max}$ is the maximum remote semantic processing rate in \eqref{R_n^M(t)}, and  $R_n^{D,\max}$ is the maximum downlink transmission rate in \eqref{R_n^D(t)}. In \eqref{A.3}, inequality $(e)$ is due to the fact that $(\max\{0,x-y\}+z)^2\leq x^2+y^2+z^2+2x(z-y)$, for all $x,y,z\geq0$, and $(f)$ holds because we scale quadratic terms to their upper bounds. 

 Besides, for virtual queue $X_n^r(t)$, we have 
 \begin{small}
 \begin{align} \label{A.5}
 	&\frac{1}{2}X_n^r(t+1)^2-\frac{1}{2}X_n^r(t)^2\nonumber\\
 	&\leq\frac{1}{2}\left(X_n^r(t)-\tau R_n^L(t)-\frac{\tau_n^U(t)R_{n}^U(t)}{\beta_n(t)}+\tau R_n^{avg}\right)^2-\frac{1}{2}X_n^r(t)^2\nonumber,\\
 	&\leq \frac{1}{2}\left(\tau R_n^{L,\max}+\frac{\tau R_{n}^{U,\max}}{\beta_n^{\min}}\right)^2+\frac{1}{2}\left(\tau R_n^{avg}\right)^2\nonumber\\
 	&-X_n^r(t)\left(\tau R_n^L(t)+\frac{\tau_n^U(t)R_{n}^U(t)}{\beta_n(t)}-\tau R_n^{avg}\right).
 \end{align}
\end{small}
 Finally, by plugging \eqref{A.1}, \eqref{A.3} \eqref{A.5} and $Q_n^{total}(t+1)$ into \eqref{Delta_V}, we can get \eqref{theorem_1_eq}. \hfill $\blacksquare$ 
\section{Proof of Corollary \ref{coro1}}
\label{}	
\setcounter{equation}{0}
\renewcommand\theequation{B.\arabic{equation}}
Let $\mu$ and $\pmb{\rho}=[\rho_1,\cdots,\rho_N]^T$ be non-negative Lagrangian multipliers corresponding to constraints (\ref{Digestion factor Optimization1}f) and (\ref{Digestion factor Optimization1}b), respectively. 
In this case, 
the optimization of $f_n^L(t)$, $r_n^U(t)$ and $p_n^D(t)$ can be done separately for individual $n$. Specifically, for each $n$, 
we solve the following problem:
\begin{small}
\begin{subequations} \label{B.2}
	\begin{align}
		&\min_{\overset{f_n^L(t),r_n^U(t),}{p_n^D(t)}}   -\frac{\tilde{W}_n^2(t)\tau}{I_n}f_n^L(t)-\tilde{W}_n^4(t)\tau_n^D(t)R_n^D(t)\nonumber\\
		&+\left(\frac{a\tilde{W}_n^1(t)}{\beta_n(t)^k}-\frac{\tilde{W}_n^2(t)}{\beta_n(t)}+\tilde{W}_n^3(t)\right)\tau_n^U(t) r_n^U(t)+\mu p_n^D(t)\nonumber\\
		&+V\left(\tau\kappa_n f_n^L(t)^3+\tau_n^U(t)\frac{\sigma^2}{h_n(t)}\left(e^{\frac{r_n^U(t)\ln 2}{B}}-1\right)+\tau_n^D(t)p_n^D(t)\right)\nonumber\\
		&+\rho_n\left(\tau\frac{f_n^L(t)}{I_n}+\frac{\tau_n^U(t)r_{n}^U(t)}{\beta_n(t)}-\frac{a\tau_n^U(t)r_n^U(t)}{\beta_n(t)^k}\right),\\
		&\textrm{\textrm{s.t.}}\quad
		p_n^D(t)\leq \frac{\sigma^2}{h_n(t)}\left(e^{\frac{Q_n^D(t)\ln 2}{B\tau_n^D(t)}}-1\right),\\
		&\quad\quad\  0\leq f_n^L(t)\leq f_n^{\max},\\
		&\quad\quad\  r_n^U(t)\leq B\log_2 \left(1+\frac{h_n(t)p_n^{\max}}{\sigma^2}\right),
	\end{align}
\end{subequations}
\end{small}We first resolve the optimal $p_n^D(t)$. Denote the objective function (\ref{B.2}a) by $J_1$. Taking the derivative of $J_1$ with respect to $p_n^D(t)$, we have 
$J_1'(p_n^D(t))=-\tilde{W}_n^4(t)\tau_n^D(t)\frac{B}{\ln 2}\frac{h_n(t)}{\sigma^2+h_n(t)p_n^D(t)}+V\tau_n^D(t)+\mu$. Setting $J_1'(p_n^D(t))$ equals to zero, we can obtain the null point  $\Omega_1=\frac{\tilde{W}_n^4(t)\tau_n^D(t)B}{\ln 2\left(V\tau_n^D(t)+\mu\right)}-\frac{\sigma^2}{h_n(t)}$. If $\Omega_1<0$, the optimal $p_n^D(t) $ is $0$, which means that TD $n$ shall not transmit data at the current time slot;  If $\Omega_1>\frac{\sigma^2}{h_n(t)}\left(e^{\frac{Q_n^D(t)\ln 2}{B\tau_n^D(t)}}-1\right)$, the optimal $p_n^D(t) $ is $\frac{\sigma^2}{h_n(t)}\left(e^{\frac{Q_n^D(t)\ln 2}{B\tau_n^D(t)}}-1\right)$; otherwise, the optimal $p_n^D(t)=\Omega_1$. Since the sub-gradient of $\mu$ of (B.1a) is $\sum_{n=1}^Np_n^D(t)-P_{MEC}$. Therefore, bisection method can be employed to find the optimal $\mu$.

Then, taking derivative of $J_1$ with respect to $f_n^L(t)$, we can get that
$J_1'(f_n^L(t))=-\frac{\tilde{W}_n^2(t)\tau}{I_n}+\frac{\rho_n\tau}{I_n}+3V\tau\kappa_nf_n^L(t)^2$. If $-\frac{\tilde{W}_n^2(t)\tau}{I_n}+\frac{\rho_n\tau}{I_n}>0$, i.e., $\rho_n-\tilde{W}_n^2(t)>0$, $J_1'(f_n^L(t))>0$, $J_1(f_n^L(t))$ increases with $f_n^L(t)$. In this case, the optimal $f_n^L(t)=0$. If $\rho_n-\tilde{W}_n^2(t)\leq 0$, we have $\Omega_2=\sqrt{\frac{\tilde{W}_n^2(t)-\rho_n}{3VI_n\kappa_n}}$. In this case, the optimal $f_n^L(t)=\min\{\Omega_2,f_n^{\max}\}$. 

Similarly, taking derivative of $J_1$ with respect to $r_n^U(t)$, we have
\begin{small}
\begin{align}
	&J_1'(r_n^U(t))=\left(\frac{a\tilde{W}_n^1(t)}{\beta_n(t)^k}-\frac{\tilde{W}_n^2(t)}{\beta_n(t)}+\tilde{W}_n^3(t)\right)\tau_n^U(t)\nonumber\\
	&+V\tau_n^U(t)\frac{\ln 2\sigma^2}{Bh_n(t)}e^{\frac{r_n^U(t)\ln2}{B}}+\rho_n\tau_n^U(t)\left(\frac{1}{\beta_n(t)}-\frac{a}{\beta_n(t)^k}\right).
\end{align}
\end{small}Denote $\omega_n(t)=\left(\frac{a\tilde{W}_n^1(t)}{\beta_n(t)^k}-\frac{\tilde{W}_n^2(t)}{\beta_n(t)}+\tilde{W}_n^3(t)\right)\tau_n^U(t)+\rho_n\tau_n^U(t)\left(\frac{1}{\beta_n(t)}-\frac{a}{\beta_n(t)^k}\right)$. If $\omega_n(t)\geq 0$, the optimal $r_n^U(t)=0$; otherwise, we have $\Omega_3=B\log_2 \frac{-\omega_n(t)Bh_n(t)}{\tau_n^U(t)\ln2V\sigma^2}$. In this case, the optimal $r_n^U(t)=\min\left\{\Omega_3, B\log_2\left(1+\frac{h_n(t)p_n^{\max}}{\sigma^2}\right)\right\}$. And the sub-gradient of $\rho_n$ is $\tau\frac{f_n^L(t)}{I_n}+\frac{\tau_n^U(t)r_{n}^U(t)}{\beta_n(t)}- Q_n^L(t)-\frac{a\tau_n^U(t)r_n^U(t)}{\beta_n(t)^k}$.  The optimal $\rho_n$ can be obtained by bisection method. The algorithm flow is summarized in Algorithm~\ref{algB}.

\section{Proof of Theorem \ref{coro2}}
\label{}	
\setcounter{equation}{0}
\renewcommand\theequation{C.\arabic{equation}}
The Lagrangian dual function of \eqref{Digestion factor Optimization2} is given by
\begin{small}
   \begin{subequations} \label{C.1}
	\begin{align}
		\hspace{-2em}\min_{\tilde \beta_n(t)}\  & \tau_n^U(t) r_n^U(t)\left(a\tilde{W}_n^1(t)\tilde \beta_n(t)^k-\tilde{W}_n^2(t)\tilde\beta_n(t)\right)\nonumber\\
		&+\xi\bigg(\tau\frac{f_n^L(t)}{I_n}+\tau_n^U(t)r_{n}^U(t)\tilde\beta_n(t) -  Q_n^L(t)-a\tau_n^U(t)r_n^U(t)\nonumber\\
	 &\times\left[k\left(\tilde{\beta}_n^{(r)}(t)\right)^{k-1}\tilde\beta_n(t)+(1-k)\left(\tilde\beta_n^{(r)}(t)\right)^k\right]\bigg),\\
	\hspace{-1em}\textrm{\textrm{s.t.}} \ 	& 1\leq \tilde\beta_n(t)\leq \frac{1}{\beta_n^{\min}},
	\end{align}
\end{subequations}
\end{small}where $\xi$ is the dual variable with respect to (\ref{Digestion factor Optimization2}b). Taking derivative of (\ref{C.1}a), denoted by $J_2$, with respect to $\tilde\beta_n(t)$, we have
$J_2'\left(\tilde\beta_n(t)\right)=\tau_n^U(t)$ $r_n^U(t) \left[k a\tilde{W}_n^1(t)\tilde \beta_n(t)^{k-1}-\tilde{W}_n^2(t)+\xi-\xi a k\left(\tilde{\beta}_n^{(r)}(t)\right)^{k-1}\right]$. If $\small \tau_n^U(t) r_n^U(t) \left[\xi-\tilde{W}_n^2(t)-\xi a k\left(\tilde{\beta}_n^{(r)}(t)\right)^{k-1}\right]\geq 0$, $J_2'$ is always non-negative for $\tilde\beta_n(t)\in[1,\frac{1}{\beta_n^{\min}}]$. Thus, the optimal $\tilde\beta_n(t) = 1$. Otherwise, $J_2'\left(\tilde\beta_n(t)\right)$ has a unique positive null point 
\begin{small}
	\begin{align}
		\Omega_4=\sqrt[k-1]{\frac{\tilde{W}_n^2(t)-\xi+\xi a k\left(\tilde{\beta}_n^{(r)}(t)\right)^{k-1}}{ka\tilde{W}_n^1(t)}}.
	\end{align}
\end{small}In this case, the optimal solution lies either in $\Omega_4$ or boundary points, i.e.,
\begin{small}
\begin{align}
	\tilde\beta_n(t)=\left\{
	\begin{aligned}
		1,\quad\quad\quad\quad \quad\quad \quad & \text{if }  \Omega_4<1;\\
		\min\left\{\Omega_4,\frac{1}{\beta_n^{\min}}\right\},\quad &\text{otherwise}.
	\end{aligned}\right.
\end{align}  
\end{small}Due to $\beta_n(t)=1/\tilde\beta_n(t)$, the optimal $\beta_n(t)$ can be obtained by \eqref{coro2eq}. And the optimal $\xi$ can be obtained by a bisection method similar to that in Algorithm~\ref{algB}.

\section{Proof of Corollary \ref{coro3}}  \label{}	
\setcounter{equation}{0}
\renewcommand\theequation{\textsc{D}.\arabic{equation}}
 
Problem \eqref{remote comp1} can be decomposed into a series of the following problems with the same structure: 
\begin{small}
\begin{subequations}\label{D.2}
	\begin{align}
		\min_{f_n^O(t)}\quad & \varphi_n(t)f_n^O(t)+V\tau\kappa_M f_n^O(t)^3+\nu f_n^O(t),\\
		\textrm{\textrm{s.t.}} \quad
		&0\leq f_n^O(t) \leq \frac{Q_n^O(t)G_n(t)I_n}{\tau},
	\end{align}
\end{subequations}
\end{small}where $\varphi_n(t)=\frac{\tilde{W}_n^4(t)H_n(t)\tau-\tilde{W}_n^3(t)\tau}{G_n(t)I_n}$, $\nu\geq 0$ is Lagrangian multiplier corresponding to constraints (\ref{remote comp1}c). Taking derivation of (\ref{D.2}a), termed as $J_3$, we can obtain that
$J_3'(f_n^O(t))=\varphi_n(t)+\nu+3V\tau\kappa_Mf_n^O(t)^2$. If $\varphi_n(t)+\nu>0$, $J_3'$ is always positive. (\ref{D.2}a) increases with $f_n^O(t)$. Thus, the optimal solution of \eqref{D.2} is $f_n^O(t)=0$. Otherwise, by setting $J_3'(f_n^O(t))$ equal to 0, we    
get $\Omega_5=\sqrt{-\frac{\varphi_n(t)+\nu}{3V\tau\kappa_M}}$. Thus, the optimal solution is $f_n^O(t)=\min\left\{\sqrt{-\frac{\varphi_n(t)+\nu}{3V\tau\kappa_M}},\frac{Q_n^O(t)G_n(t)I_n}{\tau}\right\}$. And the optimal $\nu$ can be obtained by the bisection method. 

	\bibliography{IEEEabrv,Ref}

\end{document}